\documentclass[preprints]{elsarticle}
\usepackage{latexsym}
\usepackage{amsmath}
\usepackage{amssymb}
\usepackage{wasysym} %%For \qed
\usepackage[all]{xy}

\journal{Science of Computer Programming}

\newcommand{\ol}{\overline}
\newcommand{\comp}{\mathop{\circ}}
\newcommand{\may}{\diamond}
\newcommand{\must}{\Box}
\newcommand{\tran}[1]{\stackrel{#1}{\rightarrow}}
\newcommand{\lra}{\longrightarrow}
\newcommand{\darrow}{{\downarrow}}

\newcommand{\ccsim}{\lesssim_{cc}}
\newcommand{\pbsim}{\lesssim_{B}}
\newcommand{\cv}{\mathrm{cv}}
\newcommand{\ct}{\mathrm{ct}}

\newcommand{\cc}{\mathit{cc}}
\newcommand{\mts}{\mathit{mts}}
\newcommand{\calI}{\mathcal{I}}
\newcommand{\Sign}{\mathbf{Sign}}
\newcommand{\sen}{\mathit{sen}}
\newcommand{\Mod}{\mathbf{Mod}}
\newcommand{\Set}{\mathbf{Set}}
\newcommand{\Cat}{\mathbf{Cat}}

\newcommand{\calC}{\mathcal{C}}
\newcommand{\calM}{\mathcal{M}}
\newcommand{\calMe}{\mathcal{MC}}
\newcommand{\calN}{\mathcal{N}}
\newcommand{\calT}{\mathcal{T}}

\newcommand{\mtsp}{\mathcal{T}_M}
\newcommand{\ltsp}{\mathcal{T}_L}

\newcommand{\<}{\sqsubseteq}

\begin{document}
\newtheorem{theorem}{Theorem}
\newtheorem{lemma}[theorem]{Lemma}
\newtheorem{proposition}[theorem]{Proposition}
\newtheorem{corollary}[theorem]{Corollary}
\newdefinition{definition}{Definition}
\newdefinition{remark}{Remark}
\newdefinition{example}{Example}
\newproof{proof}{Proof}

\begin{frontmatter}

\title{On the specification of modal systems: a comparison of three frameworks\tnoteref{t1}}

%% 
%%\title{Relating modal refinements, covariant-contravariant simulations
%% and partial bisimulations\tnoteref{t1}}

\tnotetext[t1]{Research supported by Spanish
 projects DESAFIOS10 TIN2009-14599-C03-01, TESIS TIN2009-14321-C02-01
 and PROMETIDOS S2009/TIC-1465, and the NILS Mobility
 Project (Abel Extraordinary Chair programme). Luca Aceto and Anna Ing\'olfsd\'ottir have been partially
 supported by the project `Processes and Modal Logics' (project
 nr.~100048021) of the Icelandic Research Fund. }

%\tnotetext[t2]{This is a revised and extended version of~\cite{AcetoEtAl10}.}

\author[ice]{Luca Aceto}
\author[ucm]{Ignacio F\'abregas}
\author[ucm]{David de Frutos-Escrig}
\author[ice]{Anna Ing\'olfsd\'ottir}
\author[ucm]{Miguel Palomino}
\address[ice]{ICE-TCS, School of Computer Science, Reykjavik University, Iceland}
\address[ucm]{Departamento de Sistemas Inform\'aticos y Computaci\'on, Universidad Complutense de Madrid, Spain}

\begin{abstract}
This paper studies the relationships between three notions of
 behavioural preorder that have been proposed in the literature:
 refinement over modal transition systems, and the
 covariant-contravariant simulation and the partial bisimulation
 preorders over labelled transition systems. It is shown that there
 are mutual translations between modal transition systems and labelled
 transition systems that preserve, and reflect, refinement and the
 covariant-contravariant simulation preorder. The translations are
 also shown to preserve the modal properties that can be expressed in
 the logics that characterize those preorders. A translation from
 labelled transition systems modulo the partial bisimulation preorder
 into the same model modulo the covariant-contravariant simulation
 preorder is also offered, together with some evidence that the former
 model is less expressive than the latter. In order to gain more
 insight into the relationships between modal transition systems
 modulo refinement and labelled transition systems modulo the
 covariant-contravariant simulation preorder, their connections are
 also phrased and studied in the context of institutions.
\end{abstract}
\end{frontmatter}

\section{Introduction}\label{Sect:intro}

{\em Modal transition systems} (MTSs) have been proposed in,
e.g.,~\cite{La89,LaTh88} as a model of reactive computation based on
states and transitions that naturally supports a notion of {\em
refinement} that is akin to the notion of implication in logical
specification languages. (See the paper~\cite{BoudolL1992} for a
thorough analysis of the connections between specifications given in
terms of MTSs and logical specifications in the setting of a modal
logic that characterizes refinement.) In an MTS, transitions come in
two flavours: the {\em may} transitions and the {\em must}
transitions, with the requirement that each must transition is also a
may transition. The idea behind the notion of refinement over MTSs is
that, in order to implement correctly a specification, an
implementation should exhibit all the transitions that are
required by the specification (these are the must transitions in the
MTS that describes the specification) and may provide the transitions
that are allowed by the specification (these are the may transitions
in the MTS that describes the specification).

The formalism of modal transition systems is intuitive, has several
variants with varying degrees of expressive power and
complexity---see, e.g., the survey paper~\cite{MTSBEATCS}---and
has recently been used as a model for the specification of
service-oriented applications. In particular, results on the
supervisory control (in the sense of Ramadge and Wonham~\cite{RW87})
of systems whose specification is given in that formalism have been
presented in, e.g.,~\cite{DaroundeauDM2010,FP2007}.

The very recent development of the notion of {\em partial
bisimulation} in the setting of labelled transition systems (LTSs)
presented in~\cite{Baetenetal,Baetenetal1} has been explicitly motivated by the
desire to develop a process-algebraic model within which one can study
topics in the field of {supervisory control}. A partial bisimulation
is a variation on the classic notion of bisimulation~\cite{Mi89,Pa81}
in which two LTSs are only required to fulfil the bisimulation
conditions on a subset $B$ of the collection of actions; transitions
labelled by actions not in $B$ are treated as in the standard
simulation preorder. Intuitively, one may think of the actions in $B$
as corresponding to the uncontrollable
events---see~\cite[page~4]{Baetenetal}. The aforementioned paper
offers a thorough development of the basic theory of partial
bisimulation.

Another recent proposal for a simulation-based behavioural relation
over LTSs, called the {\em covariant-contravariant simulation
preorder}, has been put forward in~\cite{FabregasEtAl09-orders}, and
its theory has been investigated further
in~\cite{FabregasEtAl10-logics}. This notion of simulation between
LTSs is based on considering a partition of their set of actions into
three sets: the collection of covariant actions, that of contravariant
actions and the set of bivariant actions. Intuitively, one may think
of the covariant actions as being under the control of the
specification LTS, and transitions with such actions as their label
should be simulated by any correct implementation of the
specification. On the other hand, the contravariant actions may be
considered as being under the control of the implementation (or of the
environment) and transitions with such actions as their label should
be simulated by the specification. The bivariant actions are treated
as in the classic notion of bisimulation. 

It is natural to wonder whether there are any relations among these
three formalisms. In particular, one may ask oneself whether it is
possible to offer mutual translations between specifications given in
those state-transition-based models that preserve, and reflect, the
appropriate notions of behavioural preorder as well as properties
expressed in the modal logics that accompany them---see,
e.g.,~\cite{Baetenetal,BoudolL1992,FabregasEtAl10-logics}. The aim of
this study is to offer an answer to this question. 

In this paper, we study the relationships between refinement over
 modal transition systems, and the covariant-contravariant simulation
 and the partial bisimulation preorders over labelled transition
 systems. We offer mutual translations between modal transition
 systems and labelled transition systems that preserve, and reflect,
 refinement and the covariant-contravariant simulation preorder, as
 well as the modal properties that can be expressed in the logics
 that characterize those preorders. We also give a translation from
 labelled transition systems modulo the partial bisimulation preorder
 into the same model modulo the covariant-contravariant simulation
 preorder, together with some evidence that the former model is less
 expressive than the latter. Finally, in order to gain more insight
 into the relationships between modal transition systems modulo
 refinement and labelled transition systems modulo the
 covariant-contravariant simulation preorder, we phrase and study
 their connections in the context of institutions~\cite{Inst}.

The developments in this paper indicate that the formalism of MTSs may
be seen as a common ground within which one can embed LTSs modulo the
covariant-contravariant simulation preorder or partial
bisimilarity. Moreover, there are some interesting, and non-obvious,
corollaries that one may infer from the translations we provide. See
Section~\ref{Sect:discussion}, where we use our translations to show,
e.g., that checking whether two states in an LTS are related by the
covariant-contravariant simulation preorder can always be reduced to
an equivalent check in a setting without bivariant actions, and
provide a more detailed analysis of the translations. The study of the
relative expressive power of different formalisms is, however, an art
as well as a science, and may yield different answers depending on the
conceptual framework that one adopts for the comparison. For instance,
at the level of institutions~\cite{Inst}, we provide an institution
morphism from the institution corresponding to the theory of MTSs
modulo refinement into the institution corresponding to the theory of
LTSs modulo the covariant-contravariant simulation preorder. However,
we conjecture that there is no institution morphism in the other
direction. The work presented in the study opens several interesting
avenues for future research, and settling the above conjecture is one
of a wealth of research questions we survey in
Section~\ref{Sect:future}.

The remainder of the paper is organized as
follows. Section~\ref{Sect:prelim} is devoted to preliminaries. In
particular, in that section, we provide all the necessary background
on modal and labelled transition systems, modal refinement and the
covariant-contravariant simulation preorder, and the modal logics that
characterize those preorders. In Section~\ref{cc-to-modal:sec}, we
show how one can translate LTSs modulo the covariant-contravariant
simulation preorder into MTSs modulo
refinement. Section~\ref{Sect:R2CC} presents the converse
translation. We discuss the mutual translations between LTSs and MTSs
in Section~\ref{Sect:discussion}.  As described in
Section~\ref{Sect:charforms}, the translation from MTSs and their
modal logic to the realm of LTSs modulo the covariant-contravariant
simulation preorder can be used to transfer the characteristic-formula
result from~\cite{BoudolL1992} to one for LTSs modulo the
covariant-contravariant simulation preorder.  Section~\ref{Sect:PB}
offers a translation from LTSs modulo partial bisimilarity into LTSs
modulo the covariant-contravariant simulation preorder. In
Section~\ref{Sect:institutions}, we study the relationships between
modal transition systems modulo refinement and labelled transition
systems modulo the covariant-contravariant simulation preorder in the
context of institutions. Section~\ref{Sect:future} concludes the paper
and offers a number of directions for future research that we plan to
pursue.

%%Original version
%
%This article is a substantially expanded version of the conference
%paper~\cite{AcetoEtAl10}. Apart from including the proofs of all the
%technical results, which were announced without proof in the
%conference publication with the exception of three propositions, as
%well as further remarks and explanations, the following %%scientific
%contributions are new in this version of the paper:
%\begin{itemize}
%\item the discussion of the translation $\calMe$ from Boudol-Larsen
%modal formulae to covariant-contravariant formulae presented on
%pages~\pageref{Add1-start}--\pageref{Add1-end};
%\item the discussion of the translation $\calC^{-1}$ from
%covariant-contravariant formulae to Boudol-Larsen modal formulae
%presented on page~\pageref{Add2-start};
%\item the material in Section~\ref{Sect:charforms}; and 
%\item the material on page~\pageref{instconj-start} regarding a
%conjecture from~\cite{AcetoEtAl10}.
%\end{itemize}

Compared to the preliminary version of this study presented
in~\cite{AcetoEtAl10}, besides including the full collection of proofs
of technical results, we have also added all the material on the
relationship between characteristic formulae for both MTSs and
covariant-contravariant LTSs in Section~\ref{Sect:charforms}, as well
as some additional results completing the study of the transformations
between MTSs and LTSs in the setting of institutions in
Section~\ref{Sect:institutions}.
%%on page~\pageref{instconj-start}.

\section{Preliminaries}\label{Sect:prelim}

We begin by introducing modal transition systems, with their associated notion
of (modal) refinement, and labelled transition systems modulo the
covariant-contravariant simulation preorder. For MTS and refinement we refer
the reader to, e.g.,~\cite{BoudolL1992,La89,LaTh88} for motivation and
examples, whereas 
\cite{FabregasEtAl09-orders,FabregasEtAl10-logics} can be consulted for more
information regarding covariant-contravariant simulations.

\subsection{Modal transition systems and refinement}\label{Sect:MTS}

\begin{definition}
For a set of actions $A$, a \emph{modal transition} system (MTS) is a
triple $M=(P,\tran{}_\may,\tran{}_\must)$, where $P$ is a set of states
and ${\tran{}_\may,\tran{}_\must}\subseteq {P\times A\times P}$ are
transition relations such that ${\tran{}_\must} \subseteq
{\tran{}_\may}$. 

An MTS is {\em image finite} iff the set $\{ p' \mid p \tran{a}_\may
p'\}$ is finite for each $p\in P$ and $a\in A$.
\end{definition}
The transitions in $\tran{}_\must$ are called the {\em must
transitions} and those in $\tran{}_\may$ are the {\em may
transitions}. In an MTS, each must transition is also a may
transition, which intuitively means that any required transition is
also allowed.

In what follows, we often identify an MTS, or a transition system of
any of the types that we consider in this paper, with its set of
states. In case we wish to make clear the `ambient' transition system in
which a state $p$ lives, we write $(P,p)$ to indicate that $p$ is to
be viewed as a state in $P$.

The notion of (modal) refinement $\sqsubseteq$ over MTSs is based on the idea that if $p \sqsubseteq q$
then $q$ is a `refinement' of the specification $p$. In that case,
intuitively, $q$ may be obtained from $p$ by possibly 
\begin{itemize}
\item removing some of its may transitions and/or 
\item turning some of
its may transitions into must transitions.
\end{itemize}

\begin{definition}\label{Def:modref}
A relation $R\subseteq P\times Q$ is a {\em refinement relation}
between the two modal transition systems $P$ and $Q$ if, whenever $p \mathrel{R} q$:
\begin{itemize}
\item $p\tran{a}_\must p'$ implies that there exists some $q'$ such that 
 $q\tran{a}_\must q'$ and $p' \mathrel{R} q'$;
\item $q\tran{a}_\may q'$ implies that there exists some $p'$ such that
 $p\tran{a}_\may p'$  and $p' \mathrel{R} q'$.
\end{itemize}
We write $(P,p)\< (Q,q)$ if there exists a refinement $R$ such that
$p\mathrel{R}q$. When the  MTSs are clear from the context we may
simply write $p\< q$.  

\end{definition}

\begin{example}\label{Ex:U}
Consider the MTS $U$ over the set of actions $A$ with $u$ as its only
state, and transitions $u\tran{a}_\may u$ for each $a\in A$. It is well
known, and not hard to see, that $u \sqsubseteq p$ holds for each
state $p$ in any MTS over action set $A$.  The state $u$ is often referred
to as the {\em loosest (or universal) specification}.
\end{example}

\begin{definition}\label{Def:MTSlogic}
Given a set of actions $A$, the collection of\/ {\em Boudol-Larsen's
modal formulae~\cite{BoudolL1992}} is given by the following grammar:
\[
\varphi ::= \bot \mid \top \mid \varphi\land\varphi \mid \varphi\lor\varphi\mid
            [a]\varphi \mid \langle a\rangle\varphi \qquad (a\in A).
\]  
The semantics of these formulae with respect to an MTS $P$ and a state
$p\in P$ is defined by means of the satisfaction relation $\models$,
which is the least relation satisfying the following clauses:
\begin{itemize}
\item [] $(P,p)\models \top$.
\item [] $(P,p)\models \varphi_1\land \varphi_2$ if $(P,p)\models\varphi_1$ and 
 $(P,p)\models\varphi_2$.
\item [] $(P,p) \models\varphi_1\lor \varphi_2$ if $(P,p)\models\varphi_1$ or
 $(P,p)\models\varphi_2$.
\item [] $(P,p)\models [a]\varphi$ if $(P,p')\models \varphi$ for all
 $p\tran{a}_\may p'$.
\item [] $(P,p)\models \langle a\rangle\varphi$ if $(P,p')\models \varphi$ for 
 some $p\tran{a}_\must p'$.
\end{itemize}
We say that a formula is {\em existential} if it does not contain
occurrences of $[a]$-operators, $a\in A$.
\end{definition}
For example, the state $u$ in the MTS $U$ from Example~\ref{Ex:U}
satisfies neither the formula $\langle a\rangle \top$ nor the formula
$[a]\bot$. Indeed, it is not hard to see that $(U,u)$ satisfies a
formula $\varphi$ if, and only if, $\varphi$ is a tautology.

The following result stems from~\cite{BoudolL1992}. 
\begin{proposition}\label{Prop:BLmodchar}
Let $p,q$ be states in image-finite MTSs over the set of actions
$A$. Then $p \sqsubseteq q$ iff the collection of Boudol-Larsen's
modal formulae satisfied by $p$ is included in the collection of
formulae satisfied by $q$.
\end{proposition}

\begin{remark}\label{rem:image-fin}
As is customary in the literature on modal characterizations of
bisimulation-like relations, the implication from left to right in
Proposition~\ref{Prop:BLmodchar} holds for arbitrary MTSs. On the
other hand, the implication from right to left requires the assumption
that the MTSs be image finite. See, for instance,~\cite{AcetoILS2007}
for a textbook presentation.
%%
\iffalse
  Using Definition~\ref{Def:MTSlogic} we can only obtain finite logical
  formulae. Any image-finite MTSs can be characterized by means of these, 
  as witnessed by the proof in~\cite{BoudolL1992}, but it is not possible to go
  further. 
  For instance, let $P^\omega=aP^\omega$ be the system executing infinitely
  many must $a$-transitions, $P^{<\omega}=\sum_{i\geq 0}a^i$ the system that
  can choose between any arbitrary (finite) number of must $a$-transitions, and 
  $P^{\leq\omega}=P^\omega + P^{<\omega}$ the system that may execute an
  arbitrary (possibly infinite) number of must $a$-transitions.
  There is no refinement between $P^\omega$ and $P^{\leq\omega}$, but they both
  satisfy the same modal formulae.
\fi
%%
\end{remark}

\subsection{Labelled transition systems and covariant-contravariant simulation}\label{Sect:LTS}

A labelled transition system (LTS) is just an MTS with $\tran{}_\may =
\tran{}_\must$. In what follows, we write $\rightarrow$ for the
transition relation in an LTS.

\begin{definition}\label{Def:CCsim}
Let $P$ and $Q$ be two LTSs over the set of actions $A$, and let
$\{A^r,A^l, A^{\mathit{bi}}\}$ be a partition of $A$, where the sets $A^r$, $A^l$ and $A^{\mathit{bi}}$ may be
empty\footnote{Our use of the word `partition' is therefore non-standard.}. An
{\em $(A^r,A^l)$-simulation} (or just a {\em covariant-contravariant
simulation} when the partition of the set of actions $A$ is understood
from the context) between $P$ and $Q$ is a relation $R\subseteq
P\times Q$ such that, whenever $p \mathrel{R} q$, we have:
\begin{itemize}
\item For all $a\in A^r\cup A^{\mathit{bi}}$ and all $p\tran{a} p'$, there
 exists some $q\tran{a}q'$ with $p' \mathrel{R} q'$.

\item For all $a\in A^l\cup A^{\mathit{bi}}$ and all $q\tran{a}q'$, there
 exists some $p\tran{a}p'$ with $p' \mathrel{R} q'$.
\end{itemize}
We will write $(P,p)\ccsim (Q,q)$ if there exists a
covariant-contravariant simulation $R$ such that $p \mathrel{R}
q$. Again, when the `ambient' LTSs $P$ and $Q$ are clear from the
context we may simply write $p\ccsim q$.
\end{definition}
The actions in the set $A^r$ are sometimes called {\em covariant},
those in $A^l$ are {\em contravariant} and the ones in
$A^{\mathit{bi}}$ are {\em bivariant}. When working with
covariant-contravariant simulations, we shall sometimes refer to the
triple $(A^r,A^l, A^{\mathit{bi}})$ as the {\em signature} of the
corresponding LTS, and we will say that such a system is a
covariant-contravariant LTS.
%%% We removed this since the abbreviations was only used twice. 
%%%%%%%, in short CC-LTS.

\begin{example}\label{Ex:CCex}
Assume that $a\in A^r$ and $b\in A^l$. Consider the LTS with states
$p,q,r, s$ and transitions $p \tran{a} s$, $p \tran{b} s$, $q \tran{a}
s$ and $r \tran{b} s$.
%%LTSs described by the terms $p = a+b$, $q = a$ and $r = b$. 
Then $r \ccsim p \ccsim q$, but none of the converse relations holds.
\end{example}

\begin{definition}\label{Def:formulaeCC}
{\em Covariant-contravariant modal logic} has almost the same syntax as the
one for modal refinement:
\[
\varphi ::= \bot \mid \top \mid \varphi\land\varphi \mid \varphi\lor\varphi\mid
            [b]\varphi \mid \langle a\rangle\varphi \qquad 
            (a\in A^r\cup A^\mathit{bi}, b\in A^l\cup A^\mathit{bi}).
\]  

However, the semantics differs for the modal operators, since we interpret
formulae over ordinary LTSs:
\begin{itemize}
\item [] $(P,p)\models [b]\varphi$ if $(P,p')\models \varphi$ for all
 $p\tran{b} p'$.
\item [] $(P,p)\models \langle a\rangle\varphi$ if $(P,p')\models \varphi$ for 
 some $p\tran{a} p'$.
\end{itemize}
\end{definition}
For example, both $p$ and $q$ from Example~\ref{Ex:CCex} satisfy the
formula $\langle a\rangle\top$, while $r$ does not. On the other hand,
$q$ satisfies the formula $[b]\bot$, but neither $p$ nor $r$ do.

The following result stems from~\cite{FabregasEtAl10-logics}. 
\begin{proposition}\label{Prop:modcharcc}
Let $p,q$ be states in image-finite LTSs with the same signature. Then
$p \ccsim q$ iff the collection of covariant-contravariant modal
formulae satisfied by $p$ is included in the collection of
covariant-contravariant modal formulae satisfied by $q$.
\end{proposition}
The image-finiteness requirement in the above result is there for the
same reason explained in Remark~\ref{rem:image-fin}.

\section{From covariant-contravariant simulations to modal refinements}
\label{cc-to-modal:sec}

We start our study of the connections between MTSs modulo refinement
and LTSs modulo the covariant-contravariant simulation preorder, by
showing that LTSs modulo $\ccsim$ may be translated into MTSs modulo
$\sqsubseteq$. Such a translation preserves, and reflects, those
preorders and the satisfaction of modal formulae. This result is, at
first, a bit surprising, since covariant-contravariant systems look
more expressive than modal systems because they contain three
different kinds of actions, which moreover are totally independent
from each other, while modal systems only contain two kinds of
transitions, which besides are strongly related, since any must
transition is also a may one.

The key idea behind the translation from LTSs modulo $\ccsim$ into
MTSs modulo $\sqsubseteq$ presented below is that all the transitions
in the source LTS become may transitions in its MTS
representation. Moreover, transitions in the source LTS that are labelled
with actions in $A^r \cup A^{bi}$ yield must transitions in its MTS
representation. However, a difficulty arises from this natural
construction when proving that two states $p$ and $q$ in the original
LTS are related by $\ccsim$ only if they are related by $\sqsubseteq$
in its MTS representation. Indeed, the definition of $\sqsubseteq$
requires that each may transition of $q$ be matched by an
equally-labelled may transition of $p$, whereas the definition of
$\ccsim$ only ensures this matching for transitions that have some
label in $A^l \cup A^{bi}$. We resolve this difficulty by adding
transitions of the form $r\tran{a}_\may u$ ($a\in A^r$) to each state
$r$ of the MTS representation of an LTS. Recall that, as mentioned in
Example~\ref{Ex:U}, $u$ is the only state of the loosest (or
universal) specification in the setting of MTSs modulo refinement. We
discuss this issue in more technical terms in Remark~\ref{Rem:techex}
to follow.

\begin{definition}\label{Def:cc-to-modal}
Let $P$ be a covariant-contravariant LTS with signature $\{A^r,A^l,A^\mathit{bi}\}$. 
Then the associated MTS $\calM(P)$ is constructed as follows:
\begin{itemize}
\item The set of actions of $\calM(P)$ is $A = A^r \cup A^l \cup A^{bi}$.

\item The set of states of $\calM(P)$ is that of $P$ plus a new
 state $u$.

\item For each transition $p\tran{a} p'$ in $P$, add a may transition 
 $p \tran{a}_\may p'$ in $\calM(P)$.

\item For each transition $p \tran{a} p'$ in $P$ with $a\in A^r \cup A^{bi}$,
 add a must transition $p\tran{a}_\must p'$ in $\calM(P)$. 

\item For each $a$ in $A^r$ and state $p$, add the transition 
 $p\tran{a}_\may u$ to $\calM(P)$, as well as transitions 
 $u \tran{a}_\may u$ for each action $a\in A$.   
\item There are no other transitions in $\calM(P)$.
\end{itemize}
\end{definition}
The following proposition essentially states that the translation
$\calM$ is correct.

\begin{proposition}\label{cc-to-modal}
A relation $R$ is a covariant-contravariant simulation between LTSs
$P$ and $Q$ iff $\calM(R)$ is a refinement between $\calM(P)$ and
$\calM(Q)$, where $\calM(R) = R \cup \{(u,q) \mid \textrm{$q$ a state
of $\calM(Q)$}\}$.
\end{proposition}
\begin{proof} 
We prove the two implications separately. 

$(\Rightarrow)$ Assume that $R$ is a covariant-contravariant
simulation. We shall prove that $\calM(R)$ is a refinement.  

Suppose that $p \mathrel{R} q$ and $q\tran{a}_\may q'$ in
$\calM(Q)$. By the definition of $\calM(Q)$, the transition $q\tran{a}
q'$ is in $Q$.  If $a\in A^l \cup A^{bi}$, since $p \mathrel{R} q$ and
$R$ is a covariant-contravariant simulation, we have that $p\tran{a}
p'$ in $P$ for some $p'$ such that $p' \mathrel{R} q'$. By the
construction of $\calM(P)$, it holds that $p\tran{a}_\may p'$ and we
are done.  If $a\in A^r$, then $p\tran{a}_\may u$ and $u
\mathrel{\calM(R)} q'$, as required.

Assume now that $p \mathrel{R} q$ and $p\tran{a}_\must p'$ in
$\calM(P)$.  Then $p \tran{a} p'$ in $P$ with $a\in A^r \cup
A^{bi}$. As $R$ is a covariant-contravariant simulation, it follows
that $q\tran{a} q'$ in $Q$ for some $q'$ such that $p' \mathrel{R}
q'$.  Since $a\in A^r \cup A^{bi}$, there is a must transition
$q\tran{a}_\must q'$ in $\calM(Q)$, and we are done.  To finish the
proof of this implication, recall that, as shown
in~Example~\ref{Ex:U}, each state $q$ is a refinement of $u$.

$(\Leftarrow)$ Assume that $\calM(R)$ is a refinement.  We shall prove
that $R$ is a covariant-contravariant simulation.

Suppose that
$p \mathrel{R} q$ and $q\tran{a} q'$ in $Q$ with $a\in A^l \cup
A^{bi}$.  Then $q\tran{a}_\may q'$ in $\calM(Q)$.  Since $\calM(R)$ is
a refinement, in $\calM(P)$ we have that $p \tran{a}_\may p'$ for some
$p'$ (different from $u$, because $a\notin A^r$) such that $p'
\mathrel{R} q'$.  By the construction of $\calM(P)$, it follows that
$p\tran{a}p'$ in $P$ and we are done.  

Suppose now that $p \mathrel{R}
q$ and $p\tran{a} p'$ in $P$ with $a\in A^r\cup A^{bi}$.  Then
$p\tran{a}_\must p'$ in $\calM(P)$.  Since $\calM(R)$ is a refinement,
there is some $q'$ (again, different from $u$) such that
$q\tran{a}_\must q'$ in $\calM(Q)$ and $p' \mathrel{R} q'$.  By the 
construction of $\calM(Q)$, it follows that $q \tran{a} q'$ in $Q$ and
we are done.  \qed
\end{proof}

\begin{remark}\label{Rem:techex}
As witnessed by the proof of the above proposition, the role of the
 transitions $p\tran{a}_\may u$ in $\calM(P)$ with $a\in A^r$, where
 $u$ is the loosest specification from Example~\ref{Ex:U}, is to
 satisfy `for free' the proof obligations that are generated, in the
 setting of modal refinement, by representing $A^r$-labelled
 transitions in an LTS $P$ by means of must transitions in $\calM(P)$.
 This is in the spirit of the developments in~\cite{HughesJ04}, where
 the standard simulation preorder is cast in a coalgebraic framework
 by phrasing it in the setting of bisimilarity. The coalgebraic
 recasting of simulation as a bisimulation is done in such a way that
 the added proof obligations that are present in the
 definition of bisimilarity are automatically satisfied.
\end{remark}

\begin{corollary}\label{Cor:CC2R-correct}
Let $P$ and $Q$ be two LTSs with the same signature, and let $p\in P$
and $q\in Q$. Then $(P,p) \ccsim (Q,q)$ iff $(\calM(P), p) \sqsubseteq
(\calM(Q), q)$. 
\end{corollary}

\begin{definition}
Let us extend $\calM$ to translate formulae over the modal logic that
characterizes the covariant-contravariant simulation preorder to the
modal logic for modal transition systems by simply defining
$\calM(\varphi) = \varphi$.
\end{definition}

\begin{proposition}\label{cc-to-modal:logic}
If $P$ is an LTS and $\varphi$ is a formula of the logic that characterizes
covariant-contravariant simulation, then for each $p\in P$:
\[
(P,p)\models \varphi \iff (\calM(P),p)\models \calM(\varphi).
\]  
\end{proposition}
\begin{proof}
By structural induction on $\varphi$.  The only non-trivial cases are
the ones corresponding to the modal operators, which we detail
below. (In all the following proofs, the steps labelled `IH' are those
that use the induction hypothesis.)
\begin{itemize}
\item $\langle a\rangle\varphi$, with $a\in A^r\cup A^{bi}$.
\begin{eqnarray*}
(P,p)\models\langle a\rangle\varphi 
&\iff &\textrm{there is $p\tran{a} p'$ in $P$ with $(P,p')\models\varphi$}\\
&\stackrel{\mathrm{IH}}{\iff} &\textrm{there is $p\tran{a}_\must p'$ in $\calM(P)$
 with $(\calM(P),p')\models \calM(\varphi)$}\\
&\iff &(\calM(P),p) \models \langle a\rangle \calM(\varphi)\\
&\iff &(\calM(P),p) \models \calM(\langle a\rangle  \varphi)
\end{eqnarray*}

\item $[a]\varphi$, with $a\in A^l\cup A^{bi}$.
\begin{eqnarray*}
(P,p)\models [a]\varphi 
&\iff &\textrm{$(P,p')\models\varphi$ for all $p\tran{a} p'$ in $P$}\\
&\stackrel{\mathrm{IH}}{\iff} &\textrm{$(\calM(P),p')\models \calM(\varphi)$ 
 for all $p\tran{a}_\may p'$ in $\calM(P)$}\\
&&\textrm{(note that $p\tran{a}_\may u$ only for $a\in A^r$)}\\
&\iff &(\calM(P),p) \models [a]\calM(\varphi)\\
&\iff &(\calM(P),p) \models \calM([a]\varphi) 
\end{eqnarray*}
\end{itemize}
\phantom{This completes the proof.} \qed
\end{proof}

\label{Add1-start}
It is natural to wonder whether it is possible to provide a version of
Proposition~\ref{cc-to-modal:logic} for formulae in Boudol-Larsen
modal logic. In particular, it would be interesting to characterize
the collections of formulae in Boudol-Larsen modal logic whose
satisfaction is preserved by $\calM$, in a suitable technical
sense. In order to address this question, let $\{A^r, A^l,
A^\mathit{bi}\}$ be the signature of some LTS $P$ and let $A=A^r\cup
A^l\cup A^\mathit{bi}$.

Define the transformation $\calMe$ from Boudol-Larsen formulae over
$A$ to covariant-contravariant formulae over the signature $\{A^r,
A^l, A^\mathit{bi}\}$ as follows:

\begin{definition}
$\calMe$ is the unique homomorphism that acts like the identity
function over $\bot$ and $\top$, and that satisfies:
\begin{itemize}
\item $\calMe(\langle a\rangle\varphi)=\left\{ 
        \begin{array}{ll}
		\langle a\rangle\calMe(\varphi) & \mathrm{if\ } a\in A^r\cup A^\mathit{bi} \\
		 \bot     & \mathrm{otherwise}
	    \end{array}
	   \right.$
	   
\item $\calMe([a]\varphi)=\left\{ 
        \begin{array}{ll}
		 [a]\calMe(\varphi) & \mathrm{if\ } a\in A^l\cup A^\mathit{bi} \\
		 \top     & \mathrm{otherwise}
	    \end{array}
	   \right.$

%%\item $\calMe(\varphi)=
%%\begin{array}{ll}
%%\varphi & \mathrm{if\ } \varphi  \mathrm{\ is\ not\ a\ modal\ formula}
%%\end{array}
%%$
  
\end{itemize}
\end{definition}
Note that $\calMe(\varphi)= \varphi$ when $\varphi$ does not contain
any modal operator.

The interplay between the transformation function $\calM$ between LTSs and MTSs, and the function $\calMe$ operating on Boudol-Larsen formulae is fully described by the following results.

\begin{proposition}\label{Prop:Add1}
Let $P$ an LTS over signature $\{A^r, A^l, A^\mathit{bi}\}$ and let $p\in P$. Suppose that $\varphi$ is a formula in Boudol-Larsen modal logic over $A=A^r\cup A^l\cup A^\mathit{bi}$. Then the following statements hold:
\begin{enumerate}
\item  If $(\calM(P),p)\models\varphi$ then $(P,p)\models\calMe(\varphi)$.

\item  If $(P,p)\models\calMe(\varphi)$ and (either $\varphi$ is existential or $A^r=\emptyset$) then $(\calM(P),p)\models\varphi$.
\end{enumerate}
\end{proposition}

\begin{proof}
We prove the two statements separately.

\begin{enumerate}
\item We proceed by induction on the structure of $\varphi$ and focus on the cases involving the modal operators.

\begin{itemize}
\item Case $\varphi=\langle a\rangle\varphi'$. Assume that $(\calM(P),p)\models\langle a\rangle\varphi'$. This means that there is some $p'$ such that $p\tran{a}_\must p'$ in $\calM(P)$ and $(\calM(P),p')\models\varphi'$. By the definition of $\calM$, $a\in A^r\cup A^\mathit{bi}$ and $p\tran{a}p'$. Moreover, by the inductive hypothesis, $(P,p)\models\calMe(\varphi')$. Therefore, $(P,p)\models\langle a\rangle\calMe(\varphi')$, and since $a\in A^r\cup A^\mathit{bi}$,  $(P,p)\models\calMe(\langle a\rangle\varphi')$.

\item Case $\varphi=[a]\varphi'$. Assume that $(\calM(P),p)\models[a]\varphi'$. If $a\in A^r$ then there is nothing to prove, since $\calMe(\varphi)=\top$. Assume therefore that $a\in A^l\cup A^\mathit{bi}$. We will prove that $(p,p)\models [a]\calMe(\varphi')$. To this end, suppose that $p\tran{a} p'$ in $P$ with $a\in A^l\cup A^\mathit{bi}$. By the definition of $\calM$ we have that $p\tran{a}_\may p'$ in $\calM(P)$. Since $(\calM(P),p)\models [a]\varphi'$, it follows that $(\calM(P),p')\models\varphi'$. The inductive hypothesis yields $(P,p')\models\calMe(\varphi')$, which was to be shown.
\end{itemize}

\item Assume that $(P,p)\models\calMe(\varphi)$ and that either $\varphi$ is existential or $A^r=\emptyset$. We show that $(\calM(P),p)\models\varphi$ by induction on the structure of $\varphi$. Again, the only interesting cases are those dealing with the modal operators.

\begin{itemize}
\item Case $\varphi=\langle a\rangle\varphi'$. Since $(P,p)\models\calMe(\varphi)$, we have that $a\in A^r\cup A^\mathit{bi}$ and that $p\tran{a}p'$ for some $p'$ such that $(P,p')\models\calMe(\varphi')$. Since $\varphi'$ is either existential or $A^r=\emptyset$, we may apply the inductive hypothesis to infer that $(\calM(P),p')\models\varphi'$. By the definition of $\calM$, we have that $p\tran{a}_\must p'$. Therefore $(\calM(P),p)\models\langle a\rangle\varphi'$, which was to be shown.

\item Case $\varphi=[a]\varphi'$. Since $\varphi$ is not existential, we have that $A^r=\emptyset$. So $\calMe([a]\varphi')=[a]\calMe(\varphi')$ and $(P,p)\models [a]\calMe(\varphi')$ by assumption. 

Let $p\tran{a}_\may p'$ in $\calM(P)$. As $a\in A^r\cup A^\mathit{bi}$, it follows that $p\tran{a}p'$ in $P$. Therefore $(P,p')\models\calMe(\varphi')$. By induction, $(\calM(P),p')\models\varphi'$. Since $p\tran{a}_\may p'$ was chosen arbitrarily, it follows that $(\calM(P),p)\models [a]\varphi'$, and we are done.\qed
\end{itemize}
\end{enumerate}
\end{proof}

\begin{remark}
The proviso that $\varphi$ is existential or $A^r=\emptyset$ is
necessary in statement 2 of the above proposition. To see this, assume
that $a\in A^r$ and consider the Boudol-Larsen formula $[a]\bot$. Then
the LTS with $0$ as its only state and no transitions satisfies
$\top=\calMe([a]\bot)$. On the other hand, $0\tran{a}_\may u$ holds in
$\calM(0)$, and therefore $(\calM(0),0)\not\models [a]\bot$. This
point is related to some observations we shall present in
Section~\ref{Sect:institutions}.
\end{remark}

From the addition of the sink state $u$ when defining $\calM$ it follows that
the may transitions corresponding to the covariant actions $a\in A^r$ play no
role when comparing the transformation of two covariant-contravariant
LTSs. Unfortunately, in this way the transformation $\calMe$ does not preserve
the satisfaction of formulae in all the cases, as shown by the counterexample
above. 

\begin{remark}[Open question] 
Is there a (compositional) translation $\mathcal{T}$ from Boudol-Larsen logic to
covariant-contravariant modal logic such that
\[(P,p)\models\mathcal{T}(\varphi)\;\textrm{implies}\; (\calM(P),p)\models\varphi
\]
for all LTSs $P$, states $p\in P$ and formulae $\varphi$? 
\end{remark}

%%\begin{remark} Our feeling is that the answer is negative, but we 
%%have not given this matter much thought.
%%\end{remark}
\label{Add1-end}

\section{From modal refinements to covariant-contravariant 
simulations}\label{Sect:R2CC}

We next show that MTSs modulo $\sqsubseteq$ may be translated into
LTSs modulo $\ccsim$. As the one studied in the previous section, our
translation preserves, and reflects, those preorders and the
satisfaction of modal formulae. This is, to our mind, a less
surprising result than the one presented in the previous section, even
if in order to obtain it we have to introduce two ``copies'' of each
action $a\in A$: one covariant $cv(a)\in A^r$ to represent must
transitions, and another contravariant $ct(a)\in A^l$ to represent may
transitions. As a matter of fact, we do not need the additional
generality that is offered by the possibility of also having bivariant
actions in the signature to adequately represent any MTS.

\begin{definition}\label{Def:R2CC}
Let $M$ be an MTS with set of actions $A$.  The LTS $\calC(M)$, with
signature $A^r = \{ \cv(a) \mid a\in A\}$, $A^l =\{ \ct(a) \mid
a\in A\}$ and $A^\mathit{bi}=\emptyset$, is constructed as follows:
\begin{itemize}
\item The set of states of $\calC(M)$ is the same as that of $M$.
\item For each transition $p\tran{a}_\may p'$ in $M$, add  
 $p\tran{\ct(a)} p'$ to $\calC(M)$.
\item For each transition $p\tran{a}_\must p'$ in $M$, add
 $p\tran{\cv(a)} p'$ to $\calC(M)$.
\item There are no other transitions in $\calC(M)$.
\end{itemize}
\end{definition}
Observe that the LTSs obtained as a translation of an MTS have the
following properties: 
\begin{enumerate}
\item $A^\mathit{bi}=\emptyset$ and 
\item there is a
bijection $h: A^r \rightarrow A^l$ such that if $p\tran{a}p'$ with
$a\in A^r$ then $p\tran{h(a)} p'$.
\end{enumerate}
The latter requirement corresponds to the fact that each must
transition in an MTS is also a may transition. 

%%% We reinstated the original text. 
\iffalse
Moreover, there is a
bijection $h: A^r \rightarrow A^l$ such that if $p\tran{a}p'$ with
$a\in A^r$ then $p\tran{h(a)} p'$. This relation is encoding the constraint than may and must transitions satisfies in order to have a MTS.

The latter requirement corresponds to the fact that each must
transition in an MTS is also a may transition. 
\fi
%%%

The following proposition states that the translation
$\calC$ is correct.

\begin{proposition}\label{modal-to-cc}
A relation $R$ is a refinement between $P$ and $Q$ iff $R$ is a
covariant-contravariant simulation between $\calC(P)$ and $\calC(Q)$.
\end{proposition}
\begin{proof}
We prove the two implications separately. 

$(\Rightarrow)$ Assume that $p \mathrel{R} q$.  If $p\tran{\cv(a)} p'$
in $\calC(P)$ then, by construction, $p\tran{a}_\must p'$ in $P$.
Since $R$ is a refinement, there is some $q'$ in $Q$ with
$q\tran{a}_\must q'$ and $p' \mathrel{R} q'$. Since $q\tran{\cv(a)}
q'$ is in $\calC(Q)$ by construction, we are done.  Now, assume that
$q\tran{\ct(a)} q'$ in $\calC(Q)$.  Then $q\tran{a}_\may q'$ in $Q$
and, since $R$ is a refinement, $p\tran{a}_\may p'$ in $P$ for some $p'$ 
with $p' \mathrel{R} q'$.  By construction, $p\tran{\ct(a)} p'$ is 
in $\calC(P)$ and we are done.

$(\Leftarrow)$ Assume that $p \mathrel{R} q$.  If $q\tran{a}_\may q'$
in $Q$ then $q\tran{\ct(a)} q'$ in $\calC(Q)$ and, since $R$ is a
covariant-contravariant simulation, $p\tran{\ct(a)} p'$ for some $p'$
in $\calC(P)$ such that $p' R q'$; hence $p\tran{a}_\may p'$ in $P$ as
required.  Now, if $p\tran{a}_\must p'$ in $P$ then $p\tran{\cv(a)}
p'$ in $\calC(P)$.  Since $R$ is a covariant-contravariant simulation,
there is some $q'$ in $\calC(Q)$ with $q\tran{\cv(a)} q'$ and $p'
\mathrel{R} q'$, and therefore $q\tran{a}_\must q'$ in $Q$.  \qed
\end{proof}

\begin{corollary}\label{Cor:R2CC-correct}
Let $P$ and $Q$ be two MTSs with the same action set, and let $p\in P$
and $q\in Q$. Then $(P,p) \sqsubseteq (Q,q)$ iff $(\calC(P), p) \ccsim
(\calC(Q), q)$. 
\end{corollary}

\begin{remark}\label{Rem:injC}
It is easy to see that the mapping $\calC$ is injective. Therefore,
given an LTS $P$ that is in the range of $\calC$, we may write
$\calC^{-1}(P)$ for the unique MTS whose $\calC$-image is $P$.
\end{remark}

Again, we can also extend the translation $\calC$ to also translate
modal formulae. However, in this case, the change of alphabet requires
a simple, but non-trivial, definition of the extension.

\begin{definition}
Let us extend $\calC$ to translate formulae over the modal logic for
modal transition systems with set of actions $A$ to the modal logic
that characterizes covariant-contravariant simulation with signature
$A^r = \{ \cv(a) \mid a\in A\}$, $A^l =\{ \ct(a) \mid a\in A\}$ and
$A^\mathit{bi}=\emptyset$.
\begin{itemize}
\item $\calC(\bot) = \bot$.
\item $\calC(\top) = \top$.
\item $\calC(\varphi\land\psi) = \calC(\varphi)\land \calC(\psi)$.
\item $\calC(\varphi\lor\psi) = \calC(\varphi)\lor \calC(\psi)$.
\item $\calC(\langle a\rangle\varphi) = \langle \cv(a)\rangle \calC(\varphi)$.
\item $\calC([a]\varphi) = [\ct(a)]\calC(\varphi)$.
\end{itemize}
\end{definition}

\begin{proposition}\label{modal-to-cc:logic}
If $P$ is an MTS and $\varphi$ is a Boudol-Larsen modal formula, then
for each $p\in P$:
\[
(P,p)\models \varphi \iff (\calC(P),p)\models \calC(\varphi).
\]  
\end{proposition}
\begin{proof}
By structural induction on $\varphi$, with the only non-trivial cases being
those that correspond to the modal operators:
\begin{itemize}
\item $[a]\varphi$, with $a\in A$.
\begin{eqnarray*}
(P,p)\models [a]\varphi
&\iff& \textrm{$(P,p')\models \varphi$ for all $p\tran{a}_\may p'$ in $P$} \\
&\stackrel{\mathrm{IH}}{\iff}& \textrm{$(\calC(P),p')\models \calC(\varphi)$
 for all $p\tran{\ct(a)} p'$ in $\calC(P)$}\\
&\iff& (\calC(P),p)\models [\ct(a)]\calC(\varphi)
\end{eqnarray*}

\item $\langle a\rangle\varphi$, with $a\in A$.
\begin{eqnarray*}
(P,p)\models \langle a\rangle\varphi
&\iff& \textrm{$(P,p')\models \varphi$ for some $p\tran{a}_\must p'$ in $P$} \\
&\stackrel{\mathrm{IH}}{\iff}& \textrm{$(\calC(P),p')\models \calC(\varphi)$
 for some $p\tran{\cv(a)} p'$ in $\calC(P)$}\\
&\iff& (\calC(P),p)\models \langle\cv(a)\rangle\calC(\varphi)
\end{eqnarray*}
\qed
\end{itemize}
%% \qed
\end{proof}

\begin{remark}
In fact, it is very easy to see that the translations $\calM$ and
$\calC$ also preserve, and reflect, the satisfaction of formulae in
the extensions of the logics from Definitions~\ref{Def:MTSlogic}
and~\ref{Def:formulaeCC} with infinite conjunctions and disjunctions.
\end{remark}

\label{Add2-start}
It is natural to wonder whether it is possible to provide a version of
Proposition~\ref{modal-to-cc:logic} for formulae in
covariant-contravariant modal logic over the signature $A^r = \{
\cv(a) \mid a\in A\}$, $A^l =\{ \ct(a) \mid a\in A\}$ and
$A^\mathit{bi}=\emptyset$. To this end, let $\calC^{-1}$ denote the
inverse of $\calC$ over Boudol-Larsen modal formulae defined in the
obvious way. We then have that:

\begin{proposition}\label{prop:Add2}
Let $P$ be an MTS over the set of actions $A$, and let $\varphi$ be a
covariant-contravariant modal formula over the signature $A^r = \{
\cv(a) \mid a\in A\}$, $A^l =\{ \ct(a) \mid a\in A\}$ and
$A^\mathit{bi}=\emptyset$. Then, for each $p\in P$.
\[
(P,p) \models \calC^{-1}(\varphi) \iff (\calC(P),p) \models \varphi . 
\]
\end{proposition}
\begin{proof}
By Proposition~\ref{modal-to-cc:logic}, 
\[
(P,p) \models \calC^{-1}(\varphi) \iff (\calC(P),p) \models \calC(\calC^{-1}(\varphi)) . 
\]
The claim now follows since $\calC(\calC^{-1}(\varphi)) = \varphi$.
\qed
\end{proof}

The above observation is in contrast with the result we established
earlier in Proposition~\ref{Prop:Add1}. This may be taken to be a
first indication that the translation from MTSs to LTSs, and the
accompanying one for the associated modal logics, is ``more natural''
than the one from LTSs to MTSs provided in
Section~\ref{cc-to-modal:sec}. We will explore this issue in more
detail in Section~\ref{Sect:institutions}.
\label{Add2-end}

%%%%%%% Luca: I have omitted the additional result
%%%%%%% We can reinstate in a full version
%%%%%%% Omitted again. It might be best to prove it after Proposition 21,
%%%%%%% if it holds. 
\iffalse
An additional result.

\begin{lemma}
If $\varphi\leq\psi$ then $M(\varphi)\leq M(\psi)$.  
\end{lemma}
\begin{proof}[sketch]
Rename $a\in A^r$ as $\cv(a)$ and $b\in A^l$ as $\ct(b)$ and use
the equivalence proved in Proposition~\ref{inst-morphism:prop}.
\qed
\end{proof}
\fi
%%%%%%%%%%%%%%

\section{Discussion of the previous translations}\label{Sect:discussion}

In Sections~\ref{cc-to-modal:sec} and \ref{Sect:R2CC}, we saw that it is
possible to translate back and forth between the world of LTSs modulo
the covariant-contravariant simulation preorder and MTSs modulo
refinement. The translations we have presented preserve, and reflect,
the preorders and the relevant modal formulae. There are, however,
some interesting, and non-obvious, corollaries that one may infer from
the translations.

To begin with, assume that $P$ and $Q$
are two LTSs with the same signature, with $A^\mathit{bi}\neq
\emptyset$. Let $p\in P$ and $q\in Q$ be such that $(P,p) \ccsim (Q,q)$. By
Corollary~\ref{Cor:CC2R-correct}, we know that this holds exactly when
$(\calM(P), p) \sqsubseteq (\calM(Q), q)$. Using
Corollary~\ref{Cor:R2CC-correct}, we therefore have that checking
whether $(P,p) \ccsim (Q,q)$ is equivalent to verifying whether
$(\calC(\calM(P)), p) \ccsim (\calC(\calM(Q)), q)$. Note now that
$A^\mathit{bi}$ is empty in the signature for the LTSs
$\calC(\calM(P))$ and $\calC(\calM(Q))$. Therefore, checking whether
two states are related by the covariant-contravariant simulation
preorder can always be reduced to an equivalent check in a setting
without bivariant actions.

It is also natural to wonder whether there is any relation between a
state $p$ in an LTS $P$ and the equally-named state in
$\calC(\calM(P))$. Similarly, one may wonder whether there is any
relation between a state $p$ in an MTS $P$ and the equally-named state
in $\calM(\calC(P))$. In both cases, we are faced with the difficulty
arising from the fact that the transition systems resulting from the
compositions of the two translations are over the alphabet $\{
\cv(a),\ct(a) \mid a\in A\}$, whereas the original system $P$ had
transitions labelled by actions in $A$.  In order to overcome this
difficulty, we consider the renaming $\rho: \{ \cv(a),\ct(a) \mid a\in
A\} \rightarrow A$ that maps both $\cv(a)$ and $\ct(a)$ to $a$, for
each $a\in A$. Besides, for any transition system $P$ over the set of
actions $\{ \cv(a),\ct(a) \mid a\in A\}$, we write $\rho(P)$ for the
transition system that is obtained from $P$ by renaming the label of
each transition in $P$ as indicated by $\rho$. Then we have the
following proposition:

\begin{proposition}\label{Prop:approx}
\quad
\begin{enumerate}
\item Let $P$ be an MTS and $p\in P$. Then we have $(\rho(\calM(\calC(P))),p)
\sqsubseteq (P,p)$.
\item Let $P$ be an LTS and $p\in P$. Then we have $(P,p) \ccsim
(\rho(\calC(\calM(P))),p)$.
\item In general, $(P,p) \sqsubseteq (\rho(\calM(\calC(P))),p)$ does
not hold for an arbitrary MTS $P$ and any state $p\in P$; nor does
$(\rho(\calC(\calM(P))),p) \ccsim (P,p)$ hold, for an arbitrary LTS $P$ and any state
$p\in P$.
\end{enumerate}
\end{proposition}
\begin{proof}
We limit ourselves to detailing the proof for the second statement and
to offering counter-examples proving the third one. The proof of the
first claim follows similar lines to the one for the second, and in fact is
even simpler.

In order to prove the second claim, it suffices to show that the
identity relation over $P$ is a covariant-contravariant simulation
between $P$ and $\rho(\calC(\calM(P)))$. To this end, assume first
that $p \tran{a} p'$ in $P$ for some $a\in A^r\cup
A^{\mathit{bi}}$. Then $p\tran{a}_\must p'$ in $\calM(P)$. Therefore,
$p\tran{\cv(a)} p'$ in $\calC(\calM(P))$ and $p \tran{a} p'$ in
$\rho(\calC(\calM(P)))$.

Assume now that $p \tran{a} p'$ in $\rho(\calC(\calM(P)))$ for some
$a\in A^l\cup A^{\mathit{bi}}$. This means that either $p\tran{\cv(a)}
p'$ or $p\tran{\ct(a)} p'$ in $\calC(\calM(P))$. We consider these two
possibilities separately.
\begin{itemize}
\item Suppose that $p\tran{\cv(a)} p'$ in $\calC(\calM(P))$. Then
$p\tran{a}_\must p'$ in $\calM(P)$. This means that $p \tran{a} p'$ in
$P$ and $a\in A^r\cup A^{\mathit{bi}}$. By our assumption, it must be
the case that $a\in A^{\mathit{bi}}$, and we are done.
\item Suppose that $p\tran{\ct(a)} p'$ in $\calC(\calM(P))$. Then
$p\tran{a}_\may p'$ in $\calM(P)$. Since $a\in A^l\cup
A^{\mathit{bi}}$ by our assumption, we have that $p'\neq u$ in
$\calM(P)$, because $u$ can only be reached via $A^r$-labelled
may transitions. Therefore, $p'\in P$ and $p\tran{a} p'$.
\end{itemize}
This completes the proof of the second claim. 

We now argue that, in general, $(P,p) \sqsubseteq
(\rho(\calM(\calC(P))),p)$ does not hold for an MTS $P$ and a state
$p\in P$. Let $P$ be the MTS over the alphabet $A=\{a\}$, with $p$ as its
only state and with no transitions. State $p$ has an outgoing
$a$-labelled may transition in $\rho(\calM(\calC(P)))$, which cannot be
matched by $p$ in $P$. Therefore, $(P,p) \not\sqsubseteq
(\rho(\calM(\calC(P))),p)$. 

To complete the proof we now argue that, in general,
$(\rho(\calC(\calM(P))),p) \ccsim (P,p)$ does not hold for an LTS $P$
and a state $p\in P$. Let $P$ be an LTS with $A^r = \{a\}$, $p$ as its only state,
and with no transitions. The sets $A^l$ and $A^{bi}$ can be
arbitrary and play no role in the counter-example. Then it is immediate to see that state $p$ has
a transition $p \tran{a} u$ in $\rho(\calC(\calM(P)))$, but this
transition cannot be matched by $p$ in $P$.  \qed
\end{proof}

In what follows we present a result on the relationships between the
translations $\calM$ and $\calC$ for LTSs without bivariant actions.

\begin{definition}
Let $P$ be an LTS with its alphabet partitioned into $A^r$ and $A^l$.
Then the LTS $\ol{P}$ is that obtained from $P$ by simply renaming every $a\in A^r$ as
$\cv(a)$ and every $a\in A^l$ as $\ct(a)$.
\end{definition}

\begin{proposition}\label{renaming-lts}
Let $P$ be an LTS over an alphabet $A^r\cup A^l$ and let $Q$ be an MTS over the same alphabet. 
Then the following statements hold. 
\begin{enumerate}
\item If a relation $R$ is a covariant-contravariant simulation
between $\ol{P}$ and $\calC(Q)$, then $R$ is a refinement between
$\calM(P)$ and $Q$.
\item If $(\ol{P},p) \ccsim (\calC(Q), q)$ then $(\calM(P),p) \sqsubseteq
(Q,q)$, for all states $p\in P$ and $q\in Q$.
\item The converse implication of the above statement fails.
\end{enumerate}
\end{proposition}
\begin{proof} 
We limit ourselves to detailing a proof of the first statement and to
offering a counter-example showing the third. The second statement is
an immediate corollary of the first.

To prove the first statement, assume that $p \mathrel{R} q$ and that
$R$ is a covariant-contravariant simulation between $\ol{P}$ and
$\calC(Q)$.  If $q\tran{a}_\may q'$ in $Q$ then $q\tran{\ct(a)} q'$ in
$\calC(Q)$.  Since $R$ is a covariant-contravariant simulation between
$\ol{P}$ and $\calC(Q)$, there is some $p'$ in $\ol{P}$ with
$p\tran{\ct(a)} p'$ and $p' \mathrel{R} q'$.  Therefore, $p\tran{a}
p'$ in $P$ with $a\in A^l$, and $p\tran{a}_\may p'$ in $\calM(P)$ with
$p' \mathrel{R} q'$, as required.  Now, if $p\tran{a}_\must p'$ in
$\calM(P)$ then $p\tran{a} p'$ in $P$ with $a\in A^r$ and
$p\tran{\cv(a)} p'$ in $\ol{P}$.  Since $R$ is a
covariant-contravariant simulation between $\ol{P}$ and $\calC(Q)$,
there is some $q'$ in $\calC(Q)$ with $q\tran{\cv(a)} q'$ and $p'
\mathrel{R}q'$, and therefore $q\tran{a}_\must q'$ in $Q$, as required.

To see that the converse implication of the second statement in the
proposition fails in general, let $P$ be an LTS with $A^r = \{a\}$,
with $p$ as its only state and with no transitions. In this case $A^l$
can be arbitrary and plays no role in the counter-example. Let $Q$ be
a one-state MTS with the transition $q \tran{a}_\may q$. Then we have
$(\calM(P),p) \sqsubseteq (Q,q)$. On the other hand, $(\ol{P},p)
\not\ccsim (\calC(Q), q)$, because $q \tran{\ct(a)} q$ in $\calC(Q)$
and $\ct(a)$ is a contravariant action, whereas the LTS $\ol{P}$ has
no transitions. \qed
\end{proof}

\section{Characteristic formulae for processes}\label{Sect:charforms}
%%NEW SECTION

In this section, we show that the translation $\calC$ can be used to
transfer characteristic formulae from the setting of MTSs modulo
refinement to that of LTSs modulo the covariant-contravariant
simulation preorder. Characteristic formulae for processes provide an
alternative, logical characterization of a preorder in terms of {\em a
single formula}\/: given a process $t$ we obtain a formula $\chi(t)$
such that $t\< t'$ iff $t'\models\chi(t)$, for all $t'$. Compare with
Propositions~\ref{Prop:BLmodchar} and~\ref{Prop:modcharcc}, which
characterize a relation over states in terms of infinite collections
of formulae.

For consistency with the developments
in~\cite{BoudolL1992}, we focus on characteristic formulae for finite,
``essentially loop-free'' systems. Following~\cite{BoudolL1992,FabregasEtAl10-logics}, to describe these finite systems we
consider two signatures: the first generates terms describing a family
of MTSs, and the second generates terms denoting a family of covariant-contravariant LTSs.

\begin{definition}[\cite{BoudolL1992}]\label{MTSt}
Given a set of actions $A$, the set $\mtsp(A)$ of MTS process terms is given by
\[
t ::= 0 \mid \omega \mid a.t \mid a!t \mid t + t . 
\]
where $a\in A$. 

We define the `universal MTS' associated with $\mtsp(A)$ as follows:
\begin{itemize}
\item Its set of states is just $\mtsp(A)$.

\item For each term $a.t$ we have the transition $a.t\tran{a}_\may t$;
besides, for each $a\in A$, we have $\omega\tran{a}_\may\omega$.

\item For each term $a!t$, and $\mathsf{o}\in\{\must,\may\}$ we have
$a!t\tran{a}_\mathsf{o} t$.

\item For each term $t_1+t_2$, $a\in A$ and
$\mathsf{o}\in\{\must,\may\}$ we have $t_1+t_2\tran{a}_\mathsf{o}t'$,
if and only if, we have $t_i\tran{a}_\mathsf{o}t'$ for some
$i\in\{1,2\}$.
\end{itemize}
\end{definition}

Note that $\omega$ denotes the MTS $U$ from Example~\ref{Ex:U} and is
the only source of loops in the MTS we have just described. So,
abstracting from the self-loops at the leaves labelled with $\omega$,
terms in $\mtsp(A)$ may be viewed as describing finite synchronization
trees, in the sense of Milner~\cite{Mi89}. 

A term of the form $a!t$ denotes a state in an MTS that can perform an
$a$-labelled must transition, and therefore also an $a$-labelled may
transition, leading to the state described by the term $t$.

\begin{definition}
Let $(A^r,A^l, \emptyset)$ be a signature and let $A = A^r \cup
A^l$. The set $\ltsp(A)$ of LTS process terms is given by
\[
t ::= 0 \mid \omega \mid a.t \mid t + t , 
\]
where $a\in A$. 

We define the `universal LTS' associated with $\ltsp(A)$ as follows:
\begin{itemize}
\item Its set of states is just $\ltsp(A)$.

\item For each term $a.t$ we have the transition $a.t\tran{a} t$;
besides, for each $a\in A^l$, we have $\omega\tran{a}\omega$.

\item For each term $t_1+t_2$ and each $a\in A$, we have
$t_1+t_2\tran{a}t'$, if and only if, we have $t_i\tran{a}t'$ for some
$i\in\{1,2\}$.
\end{itemize}

\end{definition}

The translation $\calC$ from MTSs over the alphabet $A$ to LTSs over
the signature $(\{\cv(a)\mid a\in A\}, \{\ct(a)\mid a\in
A\},\emptyset)$ can be extended to terms in $\mtsp(A)$ yielding terms
in $\ltsp(\{\cv(a),\ct(a)\mid a\in A\})$ as the unique homomorphism
that is the identity over constants and satisfies the following
equalities:
\begin{eqnarray*}
\calC(a!t) & = & \cv(a). \calC(t) + \ct(a).\calC(t) \quad \text{and}  \\
\calC(a.t) & = & \ct(a). \calC(t) . 
\end{eqnarray*} 
%% and with the expected recursive definition for the remaining cases.

Then we have the following results:

\begin{lemma}\label{Lem:opcor}
Let $t$ be an  MTS term. Then the following statements hold:
\begin{enumerate}
\item If $t\tran{a}_\must t'$ for some MTS term $t'$ then $\calC(t)\tran{\cv(a)} \calC(t')$. 
\item If $t\tran{a}_\may t'$ for some MTS term $t'$ then $\calC(t)\tran{\ct(a)} \calC(t')$. 
\item If $\calC(t)\tran{\cv(a)} u$ for some LTS term $u$ then
$t\tran{a}_\must t'$ for some MTS term $t'$ such that $u = \calC(t')$.
\item If $\calC(t)\tran{\ct(a)} u$ for some LTS term $u$ then
$t\tran{a}_\may t'$ for some MTS term $t'$ such that $u = \calC(t')$.
\end{enumerate}
\iffalse
\[
t\tran{a}_\must t' \iff \calC(t)\tran{\cv(a)} \calC(t')
\quad\textrm{and}\quad
t\tran{a}_\may t' \iff \calC(t)\tran{\ct(a)} \calC(t').
\]
\fi
\end{lemma}
\begin{proof}
The first two statements can be proven by induction on the proof of
the relevant transition. The third and the fourth statement can be
easily shown by induction on the structure of $t$.
%%%%%%%%
\iffalse
First, if we have $t\tran{a}_\must t'$ is because we have the process term $a!t$ then, by definition, we have $cv(a)\calC{t}$. Thus we have $\calC(t)\tran{\cv(a)} \calC(t')$. On the other hand, if we have $\calC(t)\tran{\cv(a)} \calC(t')$, we have the process $cv(a)t$. By definition of $\calC$ this can only came from the process $a!t$. Hence, we have that $t\tran{a}_\must t'$.  

Now, if we have $t\tran{a}_\may t'$, we have the process term $a.t$ and, by definition, we also have $ct(a)\calC(t)$, that is, $\calC(t)\tran{\ct(a)} \calC(t')$. Analogously, if we have $\calC(t)\tran{\ct(a)} \calC(t')$, we have the term $ct(a)t$. By definition of $\calC$ this can come from $a.t$ or $a!t$, but either we get $t\tran{a}_\may t'$. 
\fi
%%%%%%%%%%%%
\qed
\end{proof}

It is not hard to see that the LTS associated with $\calC(t)$, where
$t$ is an MTS term, is the LTS one obtains by considering the MTS for
term $t$, defined as in Definition~\ref{MTSt}, and applying the
translation $\calC$ from Definition~\ref{Def:R2CC} to it. Therefore,
the following result follows essentially from
Proposition~\ref{modal-to-cc:logic}. (One can also give a simple proof
of this result using Lemma~\ref{Lem:opcor} above.)

\begin{proposition}\label{terms-formulas:prop}
For an MTS term $t$ and a modal formula $\varphi$,
\[
t\models\varphi \iff \calC(t)\models \calC(\varphi).
\]  
\end{proposition}

The above result can be used to transfer characteristic formulae for
MTS terms modulo refinement to characteristic formulae for their image
LTS terms via $\calC$. 

We begin by recalling the definition of characteristic formulae for
MTS terms modulo refinement from~\cite{BoudolL1992,La89}. 
\begin{definition}[\cite{BoudolL1992,La89}]\label{Def:charformMTS}
For each term $t\in\mtsp(A)$, the characteristic formula $\chi(t)$ is
defined as follows: 
\begin{equation}\label{Eqn:charform}
\chi(t) = \bigwedge_{\phi \in \delta(t)} \phi \wedge 
\bigwedge_{a\in A} [a] \gamma_a(t) , 
\end{equation}
where the set of formulae $\delta(t)$ and the formulae $\gamma_a(t)$
are given inductively thus
\begin{enumerate}
\item $\delta(0) = \emptyset$ and $\gamma_a(0)= \bot$, 
\item $\delta(\omega) = \emptyset$ and $\gamma_a(\omega)= \top$,
\item $\delta(a.t) = \emptyset$, $\gamma_a(a.t)= \gamma_a(t)$ and 
$\gamma_b(a.t)= \bot$ ($b\neq a$), 
\item $\delta(a!t) = \{\langle a\rangle\chi(t)\}$ and $\gamma_b(a!t)= \gamma_b(a.t)$, for each $b\in A$, and 
\item $\delta(t_1+t_2) = \delta(t_1) \cup \delta(t_2)$ and
$\gamma_a(t_1+t_2) = \gamma_a(t_1) \vee \gamma_a(t_2)$.
\end{enumerate}
As usual, an empty conjunction stands for $\top$. 
\end{definition}

The correctness of the above construction was proved by Larsen in~\cite{La89}.
\begin{proposition}\label{Prop:Larsenchar}
Let $t, t'\in \mtsp(A)$. Then $t \sqsubseteq t'$ iff $t' \models
\chi(t)$.
\end{proposition}

Note that the formula $\chi(\omega)$ is logically equivalent to
$\top$.  Moreover, for each term $t\in \mtsp(A)$, we have that, up
to logical equivalence,
\[
\bigwedge_{\phi \in \delta(t)} \phi = \bigwedge \{\langle
a\rangle\chi(t') \mid t \tran{a}_\must t' \} .
\]
Consider now the second conjunction in the formula
(\ref{Eqn:charform}). If $t$ can perform an $a$-labelled may
transition leading to a term that is equivalent to $\omega$ with
respect to the kernel of $\sqsubseteq$, then, up to logical equivalence, 
\[
[a] \gamma_a(t)   = \top . 
\]
For each term $t$, let $A_t$ be the subset of $A$ consisting of all
the actions $a$ such that each $a$-labelled may transition from $t$ leads
to a term that is {\em not} equivalent to $\omega$ with respect to the
kernel of $\sqsubseteq$. Then, up to logical equivalence, 
\[
\bigwedge_{a\in A} [a] \gamma_a(t) = \bigwedge_{a\in A_t} [a] \bigvee 
\{ \chi(t') \mid t \tran{a}_\may t' \} . 
\]
In summary, working up to logical equivalence, we can rewrite the formula
(\ref{Eqn:charform}) as follows:
\[
\bigwedge \{\langle
a\rangle\chi(t') \mid t \tran{a}_\must t' \} \vee 
 \bigwedge_{a\in A_t} [a] \bigvee 
\{ \chi(t') \mid t \tran{a}_\may t' \} .
\]

\begin{proposition}\label{terms-characteristic:prop}
$\calC(\chi(t))$ is a characteristic formula for $\calC(t)$, for each
$t\in\mtsp(A)$.
\end{proposition}
\begin{proof}
By Propositions~\ref{Prop:Larsenchar} and~\ref{terms-formulas:prop},
$\calC(t)\models \calC(\chi(t))$.  Now, assume that $s\models
\calC(\chi(t))$ for some $s\in \ltsp(\{\cv(a),ct(a)\mid a\in A\})$. We
shall show that $\calC(t)\ccsim s$. (Observe, in passing, that, since
the map $\calC$ is not surjective, the term $s$ might not be the image
of any MTS term.) To this end, it suffices to show that the relation 
\[
R = \{(\calC(t),s) \mid s \models \calC(\chi(t)), s\in
\ltsp(\{\cv(a),ct(a)\mid a\in A\}), t\in\mtsp(A)\}
\]
is a covariant-contravariant simulation. 

To see this, note, first of all, that, in the light
of the above discussion,
\[
\calC(\chi(t)) = \bigwedge \{\langle
\cv(a)\rangle\calC(\chi(t')) \mid t \tran{a}_\must t' \} \vee 
 \bigwedge_{a\in A_t} [\ct(a)] \bigvee 
\{ \calC(\chi(t')) \mid t \tran{a}_\may t' \} .
\]
The claim can now be easily shown using Lemma~\ref{Lem:opcor} and the
fact that $\calC(\omega) = \omega \ccsim s'$, for each $s'\in
\ltsp(\{\cv(a),ct(a)\mid a\in A\})$.
%%%%%%%%%%%%%%
\iffalse
Since $\calC$ is clearly injective, again by
Proposition~\ref{terms-formulas:prop} we have that $\calC^{-1}(s)
\models \chi(t)$ and therefore, since $\chi(t)$ is characteristic for
$t$, $t\sqsubseteq \calC^{-1}(s)$.  And then, because of the
properties of $\calC$ proved in Proposition~\ref{modal-to-cc},
$\calC(t)\ccsim s$ as required.  
\fi
%%%%%%%%%%%%%%
\qed
\end{proof}

This last result can be used as an alternative
to~\cite[Lemma~2]{AcetoEtAl11} to prove the existence of
characteristic formulae for LTS terms that are in the range of
$\calC$.  Indeed, for those terms, the characteristic formula derived
using the above proposition coincides with the one offered by the
direct construction given in the above-cited reference.
\iffalse
given a term $t$ over a signature with no bivariant
actions, we need first rename it to $\rho(t)$ such that every $a\in
A^r$ is replaced with $\cv(a)$ and every $a\in A^l$ with $\ct(a)$.
From Proposition~\ref{terms-characteristic:prop} we conclude that
$\calC(\chi(\calC^{-1}(\rho(t))))$ is characteristic for $\rho(t)$,
and a simple renaming is all that remains to obtain one for $t$.
\fi

\section{Partial bisimulation}\label{Sect:PB}

The partial bisimulation preorder has been proposed
in~\cite{Baetenetal} as a suitable behavioural relation over LTSs for
studying the theory of supervisory control~\cite{RW87} in a
concurrency-theoretic framework. Formally, the notion of partial
bisimulation is defined over LTSs with a set of actions $A$ and a
so-called {\em bisimulation set} $B\subseteq A$.  The LTSs considered
in~\cite{Baetenetal} also include a termination predicate $\darrow$
over states. For the sake of simplicity, and since its role is
orthogonal to our aims in this paper, instead of extending MTSs and
their refinements and/or covariant-contravariant simulations with such a
predicate, we simply omit this predicate in what follows.

\begin{definition}\label{Def:PB}
A {\em partial bisimulation with bisimulation set $B$} between two LTSs
$P$ and $Q$ is a relation $R\subseteq P\times Q$ such that, whenever
$p \mathrel{R} q$:
\begin{itemize}
\item For all $a\in A$, if $p\tran{a}p'$ then there exists some
 $q\tran{a}q'$ with $p'\mathrel{R} q'$.
\item For all $b\in B$, if $q\tran{b}q'$ then there exists some
 $p\tran{b}p'$ with $p' \mathrel{R} q'$.
\end{itemize}
We write $p\pbsim q$ if $p \mathrel{R} q$ for some partial
bisimulation with bisimulation set $B$.
\end{definition}
%%%%%%%%%%%%%%%%%%%%%%%%%%%%%%%%%%%%%%

It is easy to see that partial bisimulation with bisimulation set $B$
is a particular case of covariant-contravariant simulation.
\begin{proposition}\label{Prop:PBasCC}
Let $P$ be an LTS. A relation $R$ is a partial bisimulation with
bisimulation set $B$ iff it is a covariant-contravariant simulation
for the same LTS when it is seen as a covariant-contravariant LTS with signature $A^r = A\setminus B$, $A^l = \emptyset$ and
$A^{bi} = B$. As a consequence we have $p \pbsim q$ iff $p \ccsim q$, for each $p,q\in P$.
\end{proposition}

\begin{proof}
Immediate from the definitions.  \qed
\end{proof}

\begin{remark}
Note that, in the light of the discussion in
Section~\ref{Sect:discussion}, after having changed the signature of
the LTS $P$ in the manner described in the statement of the above result,
checking whether $p \pbsim q$ holds in $P$ can always be
reduced to verifying whether $p \ccsim q$ holds in
$\calC(\calM(P))$. This check does {\em not} involve any bivariant
action.
\end{remark}

%%%%%%%%%%%%%%%%%%%%%%%%%%%%%%%%%%%%%%%%%%%
As a corollary of the above proposition, we immediately obtain the
following result, that indicates us that, instead of the modal logic used
in~\cite{Baetenetal} to characterize the partial bisimulation preorder
with bisimulation set $B$, one can use the simpler, negation-free
logic for the covariant-contravariant simulation preorder.

\begin{corollary}
Let $p,q$ be states in some image-finite LTS. Then $p \pbsim q$ iff
the collection of formulae in Definition~\ref{Def:formulaeCC} over the
signature $A^r = A\setminus B$, $A^l = \emptyset$ and $A^{bi} = B$
satisfied by $p$ is included in the collection of formulae satisfied
by $q$.
\end{corollary}
%%%%%%%%%%%%%%%%%%%%%%%%%
Note also that, as a corollary of Proposition~\ref{Prop:PBasCC}, the
translations of LTSs and formulae defined in
Section~\ref{cc-to-modal:sec} can be applied to embed LTSs modulo the
partial bisimulation preorder into modal transition systems modulo
refinement. In this case, however, there is an easier alternative transformation
that does not require the extra state $u$.

\begin{definition}\label{Def:PB2R}
Let $P$ be an LTS over a set of actions $A$ with a bisimulation set 
$B\subseteq A$.
Then the MTS $\calN(P)$ is constructed as follows:
\begin{itemize}
\item The set of states is that of $P$.
\item For each transition $p\tran{a} p'$ in $P$, we add a transition
 $p\tran{a}_\may p'$ in $\calN(P)$.
\item For each transition $p\tran{b} p'$ in $P$ with $b\in B$, we add a transition $p\tran{b}_\must p'$ in $\calN(P)$.
\item There are no other transitions in $\calN(P)$.
\end{itemize}
\end{definition}

\begin{proposition}
$R$ is a partial bisimulation with bisimulation set $B$ between $P$ and $Q$ 
iff $R^{-1}$ is a refinement between $\calN(Q)$ and $\calN(P)$.
\end{proposition}
\begin{proof}
$(\Rightarrow)$ Assume that $R$ is a partial bisimulation with
bisimulation set $B$ and suppose that $q \mathrel{R^{-1}} p$.  If $p\tran{a}_\may
p'$ in $\calN(P)$ then $p\tran{a} p'$ in $P$. Since $R$ is a partial
bisimulation, there is some $q\tran{a} q'$ in $Q$ with $p' \mathrel{R} q'$ and, by
construction, $q\tran{a}_\may q'$ in $\calN(Q)$ with $q' \mathrel{R^{-1}} p'$.
Now, if $q\tran{a}_\must q'$ in $\calN(Q)$ then $q\tran{a} q'$ in $Q$
with $a\in B$. Since $R$ is a partial bisimulation and $p\mathrel{R} q$, there is some 
$p\tran{a} p'$ in $P$ with $p' \mathrel{R} q'$ and hence $p\tran{a}_\must p'$ in
$\calN(P)$, as required.

$(\Leftarrow)$
Analogous.
\qed
\end{proof}

\begin{remark}\label{Rem:0least}
In the special case $B=\emptyset$, the partial bisimulation preorder
is just the standard simulation preorder. Therefore, for the LTS defined by the term $0$, we have $0\pbsim p$ for each state
$p$ in any LTS $P$.  Since $B=\emptyset$, all the modal transition
systems $\calN(P)$ that result from the translation of an LTS $P$ will
have no must transitions; for such modal transition systems,
$\calN(P)\sqsubseteq 0$ always holds.  Indeed, in that case
$\sqsubseteq$ coincides with the inverse of the simulation preorder
over MTSs.
\end{remark}
%%%%%%%%%%%%%%%%%%%%%%%%%
The drawback of the direct transformation presented in
Definition~\ref{Def:PB2R}, as compared to that in
Section~\ref{cc-to-modal:sec}, is that it does not preserve the
satisfaction of modal formulae.  The problem lies in the fact that,
while the existential modality $\langle a\rangle$ allows any
transition with $a\in A$ in the partial bisimulation framework, it
requires a must transition in the setting of MTSs.

As we have seen, it is easy to express partial bisimulations as a 
special case of covariant-contravariant simulations. It is therefore natural
to wonder whether the converse also holds. We present some
indications that the partial bisimulation framework is strictly less
expressive than both modal refinements and covariant-contravariant
simulations.  

Let us assume, by way of example, that the set of actions $A$ is
partitioned into $A^r = \{a\}$ and $A^l = \{b\}$---so the set of
bivariant actions is empty. In this setting, there cannot be a
translation $\calT$ from LTSs modulo $\ccsim$ into LTSs modulo
$\pbsim$ that satisfies the following natural conditions (by abuse of
notation, we identify an LTS $P$ with a specific state $p$):
\begin{enumerate}
\item \label{eq:equivalence}
 For all $p$ and $q$, $p \ccsim q \iff \calT(p) \pbsim \calT(q)$.
\item \label{eq:homomorphism} $\calT$ is a homomorphism with respect
 to $+$, that is, $\calT(p+q) = \calT(p)+T(q)$, where $+$ denotes the
 standard notion of nondeterministic composition of LTSs from
 CCS~\cite{Mi89}. (Intuitively, this compositionality requirement
 states that the translation only uses `local information'.)
\item \label{eq:zero} There is an $n$ such that $\calT(b^n)$ is not
 simulation equivalent to $\calT(0)$, where $b^n$ denotes an LTS
 consisting of $n$ consecutive $b$-labelled transitions.
\end{enumerate}
%%%%%%%%%%%%%%%%%%%%%%%%%%%%%%%%%%
Indeed, observe that, by condition~\ref{eq:homomorphism}, 
\[
\calT(p) = \calT(p+0) = \calT(p) + \calT(0) \quad\textrm{for each $p$},  
\]
and therefore $\calT(p) + \calT(0) \pbsim \calT(p)$. 
This means that $\calT(0) \lesssim \calT(p)$ for each $p$, where $\lesssim$ is 
the simulation preorder. 
In particular, $\calT(0) \lesssim \calT(\bot)$ where $\bot$ is the process 
consisting of a $b$-labelled loop with one state, which is the least element 
with respect to $\ccsim$. 

Note now that $\bot \ccsim b^{n+1} \ccsim b^n \ccsim 0$ for each $n>0$. 
Therefore, by condition~\ref{eq:equivalence},
\[
\calT(\bot) \pbsim \calT( b^{n+1})\pbsim \calT(b^n) \pbsim \calT(0) \quad\textrm{for each $n>0$}. 
\]
Hence,
\[
\calT(\bot) \lesssim \calT(b^n)\lesssim \calT(0)\lesssim \calT(\bot) 
\quad\textrm{for each $n>0$}.
\]
This yields that, for each $n>0$, $\calT(b^n)$ is simulation equivalent to 
$\calT(0)$, which contradicts condition~\ref{eq:zero}.
(Note that we have only used the soundness of the transformation $\calT$.) 

This is clearly indicating that any $\calT$ that is compositional with
respect to $+$ and is sound, in the sense of
condition~\ref{eq:equivalence}, would have to be very odd indeed, if
it exists at all.  Modulo simulation equivalence, such a translation
would have to conflate a non-well-founded descending chain of LTSs
into a single point. 
%%, modulo simulation equivalence.

We end this section with a companion result. 
\begin{proposition}\label{Prop:noconversionintosim}
Assume that $a\in A^r$ and $b\in A^l$. Suppose furthermore that
$B=\emptyset$. Then there is no translation $\calT$ from LTSs modulo
$\ccsim$ into LTSs modulo $\pbsim$ that satisfies
conditions~\ref{eq:equivalence} and~\ref{eq:homomorphism} above.
\end{proposition}
\begin{proof}
Assume, towards a contradiction, that $\calT$ is a translation from
LTSs modulo $\ccsim$ into LTSs modulo $\pbsim$ that satisfies the
conditions in the statement of the proposition. Recall that, when $B$
is empty, $\pbsim$ is the simulation preorder (see
Remark~\ref{Rem:0least}). Therefore, using
condition~\ref{eq:homomorphism}, for each $p$ and $q$, we have that
\[
\calT(p) \pbsim  \calT(p) + \calT(q) = \calT(p+q) .
\]
This means, in particular, that $\calT(a) \pbsim \calT(a+b)$. By
condition~\ref{eq:equivalence}, it follows that $a \ccsim a+b$. This
is, however, false since $b$ is in $A^l$.  Therefore $\calT$ cannot exist. \qed
\end{proof}

\section{Institutions and institution morphisms}\label{Sect:institutions}

After defining mutual transformations between MTSs modulo refinement
and LTSs modulo the covariant-contravariant simulation preorder, we
wanted to know how close this relationship was. In particular, it is
interesting to ask whether, in a precise sense, there is a one-to-one
correspondence between those models.  Thus, in order to gain more
insight into the relationship between MTSs modulo refinement and LTSs
modulo the covariant-contravariant simulation preorder, we will now
study their connections at a more abstract level, in the context of
institutions~\cite{Inst}.  When compared at the level of institutions
it turns out that the correspondence between these models is not
one-to-one.

Admittedly, institutions and their morphisms have not been used often
in the literature to compare the expressive power of models of
concurrency. However, those notions were proposed by Burstall and Goguen
explicitly in order to relate and translate logical systems used in
computer science. We therefore think that it is interesting to use
them to study the connections between MTSs modulo refinement and LTSs
modulo the covariant-contravariant simulation preorder within a
systematic framework. An alternative, and perhaps more standard, approach
would be to study the relationships between the models we consider in
this paper following the lead of~\cite{WinskelN1995}.

To make the paper as self-contained as possible, we first provide the
formal definition of institution and some examples. In what follows,
the notation $|\mathbf{C}|$ is used to denote the set of objects in a
category $\mathbf{C}$.

\begin{definition}[\cite{Inst}]
An \emph{institution} is a quadruple $\calI= (\Sign, \sen,
\Mod,\models)$ such that
\begin{itemize}
\item $\Sign$ is a category whose objects are called signatures,
\item $\sen : \Sign \lra \Set$ is a functor associating to each signature
  $\Sigma$ a set of sentences over that signature,
\item $\Mod : \Sign^\mathrm{op} \lra\Cat$ is a functor that gives for each
  signature $\Sigma$ a category whose objects are called $\Sigma$-models,
\item $\models$ is a function associating to each $\Sigma \in |\Sign|$ a binary
  relation ${\models}_\Sigma \subseteq |\Mod(\Sigma)|\times \sen(\Sigma)$
  called $\Sigma$-satisfaction, 
\end{itemize}
so that the following \emph{satisfaction condition} holds for any $H: \Sigma
\rightarrow \Sigma'$ in $\Sign$, $M'\in |\Mod(\Sigma')|$ and all 
  $\varphi\in \sen(\Sigma)$:
\[
M'\models_{\Sigma'} \sen(H)(\varphi) \iff \Mod(H)(M')\models_{\Sigma}\varphi.
\]
\end{definition}

Institutions are particularly well-suited to capture, in a `simple'
though very abstract manner, the notion of model and, more broadly, of
logical system.  They can be used to prove general results about
logical systems without the need to consider all their possible
different instances.  For example, many-sorted equational logic can be
naturally turned into an institution as follows:
\begin{itemize}
\item $\Sign$ is the category of ``ordinary'' signatures $\Sigma=(S,\Omega)$
  composed of a set $S$ of sorts and a set $\Omega$ of many-sorted operations
  $f : s_1 \dots s_n \to s$ of varying arities.
  Morphisms are defined by a pair of functions $\mu_S: S\lra S'$ between sorts and
  $\mu_\Omega:\Omega\lra\Omega'$ between operations so that $\mu_\Omega(f) : \mu_S(s_1)\dots
  \mu_S(s_n)\to \mu_S(s)$.
\item $\sen(\Sigma)$ returns the set of all equations $(\forall X)\, t = t'$
  that can be built with the terms over $\Sigma$ and the set of sorted
  variables $X$. 
  For a mapping $H:\Sigma\lra \Sigma'$ of signatures, $\sen(H)$ simply
  translates functions in $\Sigma$ in an equation $(\forall X)\, t = t'$ to the
  corresponding functions in $\Sigma'$. 
\item $\Mod(\Sigma)$ is simply the category of $\Sigma$-algebras together
  with the morphisms between them.
  $\Mod(H)$, for $H:\Sigma \lra \Sigma'$, is the reduct mapping taking a
  $\Sigma'$-algebra $A$ to its reduct $\Sigma$-algebra, usually denoted by
  $A|_H$. 
\item Finally, the satisfaction condition can be proved to hold in many-sorted
  equational logic.
\end{itemize}
Other relevant institutions commonly used in the specialized literature
\cite{Inst,Instmorph} are those for Horn logic, first-order
logic, temporal logic, or the Common Algebraic Specification Language (CASL).

Next, we proceed to define appropriate institutions for the notions of
covariant-contravariant simulation preorder and modal transition systems.

\begin{definition}
The institution $\calI_\cc= (\Sign_\cc, \sen_\cc, \Mod_\cc,
\models_\cc)$, associated with the logic for the covariant-contravariant
simulation preorder, is defined as follows.
\begin{itemize}
\item $\Sign_\cc$ has as objects triples $(A,B,C)$ of pairwise
 disjoint sets and morphisms $f :A\cup B\cup C\lra A'\cup B'\cup C'$
 with $f(A)\subseteq A'$, $f(B)\subseteq B'$, and $f(C)\subseteq C'$.
\item $\sen_\cc(A,B,C)$ is the set of formulae in the logic
 characterizing the covariant-contravariant simulation preorder, with
 $A$ the set of covariant actions, $B$ the set of contravariant
 actions, and $C$ the set of bivariant actions.  For each signature
 morphism $f$ and formula $\varphi$, the formula $\sen(f)(\varphi)$ is
 obtained from $\varphi$ by replacing each action $a$ with $f(a)$.
\item $\Mod_\cc(A,B,C)$ is the category of LTSs over the set of
  actions $A\cup B\cup C$, with a distinguished (initial) state. In
  Section~\ref{Sect:MTS} we introduced the notation $(P,p)$ to denote
  a state $p$ inside a system $P$; here we will use the same notation
  to denote any object of $\Mod_\cc(A,B,C)$.  Then, a morphism from
  $(P,p)$ to $(Q,q)$ is a covariant-contravariant simulation $R$ such
  that $(p,q)\in R$.

 Now, if $f : A\cup B\cup C\lra A'\cup B'\cup C'$ is a signature
 morphism, then
 \[
\Mod_\cc(f): \Mod_\cc(A',B',C')\lra \Mod_\cc(A,B,C)
\]
 maps $(P,p)$ to $(P|_f,p|_f)$ and $R:P\lra Q$ to $R_f: P|_f\lra Q|_f$, where:
 \begin{itemize}
 \item The set of states of $P|_f$ is the same as that of $P$, and the
   distinguished state remains the same: $p|_f=p$.
 \item $s\tran{a} s'$ in $P|_f$ if $s\tran{f(a)} s'$ in $P$.
 \item $R|_f$ coincides with $R$.
 \end{itemize}
\item $(P,s)\models_\cc\varphi$ if $(P,s)\models \varphi$ using the notion
 of satisfaction associated with the logic for the covariant-contravariant 
 simulation preorder given in Definition~\ref{Def:formulaeCC}.
\end{itemize}
\end{definition}

\begin{proposition}\label{prop:I_cc}
$\calI_\cc$ is an institution. 
\end{proposition}
\begin{proof}
It is easy to check that all defined notions are indeed 
categories and functors.
As for the satisfaction condition, if $f:A\cup B\cup C\lra A'\cup B'\cup C'$ in
$\Sign_\cc$, $(P',s) \in \Mod_\cc(A',B',C')$, and $\varphi\in \sen_\cc(A,B,C)$, 
then
\[
(P',s)\models_\cc \sen_\cc(f)(\varphi) \iff \Mod_\cc(f)(P',s)\models_\cc \varphi
\]
can be proved by structural induction on $\varphi$. We consider the
possible forms $\varphi$ may have.
\begin{itemize}
\item $\top$ and $\bot$ are trivial.

\item For $\varphi_1\land\varphi_2$:
\begin{eqnarray*}
(P',s)\models_\cc \sen_\cc(f)(\varphi_1\land \varphi_2)
&\iff& (P',s)\models_\cc \sen_\cc(f)(\varphi_1) \land \sen_\cc(f)(\varphi_2)\\
&\stackrel{\mathrm{IH}}{\iff}& (P'|_f,s)\models_\cc \varphi_1 \textrm{ and }
       (P'|_f,s)\models_\cc \varphi_2\\
&\iff& (P'|_f,s)\models_\cc \varphi_1\land \varphi_2 . 
\end{eqnarray*}

\item Analogously for $\varphi_1\lor\varphi_2$.

\item For $\langle a\rangle\varphi$, with $a\in A\cup C$:
\begin{eqnarray*}
\lefteqn{(P',s)\models_\cc \sen_\cc(f)(\langle a\rangle\varphi) } \\
& \iff & (P',s)\models_\cc \langle f(a)\rangle\sen_\cc(f)(\varphi)\\
&\iff& \textrm{there is $s\tran{f(a)} p$ in $P'$ with 
               $(P',p)\models_\cc \sen_\cc(f)(\varphi)$}\\
&\stackrel{\mathrm{def}\ P'|_f,\ \mathrm{IH}}{\iff}& 
       \textrm{there is $s\tran{a} p$ in $P'|_f$ with 
               $(P'|_f,p)\models_\cc \varphi$}\\
&\iff& (P'|_f,s)\models_\cc \langle a\rangle\varphi . 
\end{eqnarray*}

\item For $[a]\varphi$, with $a\in B\cup C$:
\begin{eqnarray*}
\lefteqn{(P',s)\models_\cc \sen_\cc(f)([a]\varphi)}\\
&\iff& (P',s)\models_\cc [f(a)]\sen_\cc(f)(\varphi)\\
&\iff& \textrm{$(P',p)\models_\cc \sen_\cc(f)(\varphi)$
               for all $s\tran{f(a)} p$ in $P'$}\\
&\stackrel{\mathrm{def}\ P'|_f,\ \mathrm{IH}}{\iff}& 
       \textrm{$(P'|_f,p)\models_\cc \varphi$
               for all $s\tran{a} p$ in $P'|_f$}\\
&\iff& (P'|_f,s)\models_\cc  [a]\varphi . 
\end{eqnarray*}
\end{itemize}
This completes the proof. \qed
\end{proof}

In a similar way we define the institution $\calI_\mts$ for modal transition systems.

\begin{definition}
The institution $\calI_\mts= (\Sign_\mts, \sen_\mts, \Mod_\mts,
\models_\mts)$, associated with the logic for refinement over modal
transition systems, is defined as follows.
\begin{itemize}
\item $\Sign_\mts$ is the category of sets.
\item $\sen_\mts(A)$ is the set of formulae over $A$ in the logic presented in Definition~\ref{Def:MTSlogic}.
 The formula $\sen_\mts(f)(\varphi)$ is obtained from $\varphi$ by replacing 
 each action $a$ with $f(a)$.
\item $\Mod_\mts(A)$ is the category of MTSs over the set of
 labels $A$, with a  distinguished (initial) state.
 A morphism from $(M,m)$ to $(N,n)$ is a refinement $R$ such that $(m,n)\in R$.

 If $f: A\lra B$ in $\Sign_\mts$, then 
 $\Mod_\mts(f): \Mod_\mts(B)\lra \Mod_\mts(A)$ maps an MTS
 $(M,m)$ to $(M|_f, m|_f)$ and a morphism $R$ to $R|_f$, where:
 \begin{itemize}
 \item $M|_f$ has the same set of states as $M$ and the same distinguished 
   state: $m|_f=m$.
 \item $p\tran{a}_\may p'$ in $M|_f$ if $p\tran{f(a)}_\may p'$ in $M$.
 \item $p\tran{a}_\must p'$ in $M|_f$ if $p\tran{f(a)}_\must p'$ in $M$.
 \item $R|_f$ coincides with $R$.
 \end{itemize}

\item $\models_\mts$ is the notion of satisfaction presented in
Definition~\ref{Def:MTSlogic}.
\end{itemize}
\end{definition}

\begin{proposition}
$\calI_\mts$ is an institution.
\end{proposition}
\begin{proof}
Again, let us just prove the satisfaction condition
\[
(M',s)\models_\mts\sen_\mts(f)(\varphi) \iff
\Mod_\mts(f)(M',s)\models_\mts \varphi , 
\]
for $f:A\lra B$ in $\Sign_\mts$, $(M',s) \in \Mod_\mts(B)$, and
$\varphi\in \sen_\mts(A)$, by induction on $\varphi$. We consider the
possible forms $\varphi$ may have.

\begin{itemize}
\item $\top$ and $\bot$ are trivial.

\item For $\varphi_1\land\varphi_2$:
\begin{eqnarray*}
\lefteqn{(M',s)\models_\mts \sen_\mts(f)(\varphi_1\land \varphi_2)}\\
&\iff& (M',s)\models_\mts \sen_\mts(f)(\varphi_1) \land\sen_\mts(f)(\varphi_2)\\
&\stackrel{\mathrm{IH}}{\iff}& (M'|_f,s)\models_\mts \varphi_1 \textrm{ and }
       (M'|_f,s)\models_\mts \varphi_2\\
&\iff& (M'|_f,s)\models_\mts \varphi_1\land \varphi_2 . 
\end{eqnarray*}

\item Analogously for $\varphi_1\lor\varphi_2$.

\item For $\langle a\rangle\varphi$:
\begin{eqnarray*}
\lefteqn{(M',s)\models_\mts \sen_\mts(f)(\langle a\rangle\varphi)}\\
&\iff& (M',s)\models_\mts \langle f(a)\rangle\sen_\mts(f)(\varphi)\\
&\iff& \textrm{there is $s\tran{f(a)}_\must p$ in $M'$ with 
               $(M',p)\models_\mts \sen_\mts(f)(\varphi)$}\\
&\stackrel{\mathrm{def}\ M'|_f,\ \mathrm{IH}}{\iff}& 
       \textrm{there is $s\tran{a}_\must p$ in $M'|_f$ with 
               $(M'|_f,p)\models_\mts \varphi$}\\
&\iff& (M'|_f,s)\models_\mts \langle a\rangle\varphi .
\end{eqnarray*}

\item For $[a]\varphi$:

\begin{eqnarray*}
\lefteqn{(M',s)\models_\mts \sen_\mts(f)([a]\varphi)}\\
&\iff& (M',s)\models_\mts [f(a)]\sen_\mts(f)(\varphi)\\
&\iff& \textrm{$(M',p)\models_\mts \sen_\mts(f)(\varphi)$
               for all $s\tran{f(a)}_\may p$ in $M'$}\\
&\stackrel{\mathrm{def}\ M'|_f,\ \mathrm{IH}}{\iff}& 
       \textrm{$(M'|_f,p)\models_\mts \varphi$
               for all $s\tran{a}_\may p$ in $M'|_f$}\\
&\iff& (M'|_f,s)\models_\mts [a]\varphi.\phantom{aaaaaaaaaaaaaaaaaaaaaaaaa}\textrm{\qed}  
\end{eqnarray*}

\end{itemize}
\end{proof}

Having abstractly captured covariant-contravariant simulations and modal
transitions systems by means of institutions allows us to try to relate them by
means of the categorical machinery developed to that effect. 
There have been many proposals of what a morphism between institutions should
be and most are collected and discussed in \cite{Instmorph}.
Their conclusion is that there is no canonical notion that fits all situations,
but it is commonly accepted that the most natural is the one we present next.

\begin{definition}[\cite{Instmorph}]
Given institutions $\calI= (\Sign, \sen, \Mod,\models)$ and $\calI'= (\Sign',
\sen', \Mod',\models')$, an institution morphism from $\calI$ to $\calI'$
consists of a functor $\Phi: \Sign \lra \Sign'$, a natural transformation
$\beta:\Mod\Rightarrow \Mod'\comp\Phi$, and a natural transformation $\alpha:
\sen'\comp\Phi\Rightarrow\sen$, such that the condition
\[
M\models_\Sigma\alpha_\Sigma(H)  \iff \beta_\Sigma(M)\models'_{\Phi(\Sigma)} H
\]
holds for each $\Sigma\in|\Sign|$, $M\in|\Mod(\Sigma)|$, and
$H\in\sen'(\Phi(\Sigma))$. 
\end{definition}
The intuition behind these for institution morphisms is that they are ``truth
preserving'' translations from one logical system into another.
As the following result shows, one can indeed translate $\calI_\mts$ into
$\calI_\cc$ using an institution morphism. 

\begin{proposition}\label{inst-morphism:prop}
$(\Phi,\alpha,\beta) : \calI_\mts\lra\calI_\cc$ is an institution morphism,
where:
\begin{itemize}
\item $\Phi: \Sign_\mts\lra\Sign_\cc$ maps $A$ to the triple  
 $(\cv(A),\ct(A),\emptyset)$, with:
 \begin{itemize}
 \item $\cv(A) = \{\cv(a)\mid a\in A\}$ and 
 \item $\ct(A) = \{\ct(a)\mid a\in A\}$.
 \end{itemize}
 For $f:A\lra B$, we define $\Phi(f)(\cv(a)) = \cv(f(a))$ and
 $\Phi(f)(\ct(a)) = \ct(f(a))$. 

\item The natural transformation $\alpha:\sen_\cc\comp\Phi\Rightarrow\sen_\mts$
 translates a formula $\varphi$ in $\sen_\cc(\cv(A),\ct(A),\emptyset)$ as 
 follows:
 \begin{itemize}
 \item $\alpha(\top) =\top$, $\alpha(\bot) = \bot$.
 \item $\alpha(\varphi_1\land\varphi_2) = 
        \alpha(\varphi_1)\land\alpha(\varphi_2)$.
 \item $\alpha(\varphi_1\lor\varphi_2) = 
        \alpha(\varphi_1)\lor\alpha(\varphi_2)$.
 \item $\alpha(\langle \cv(a)\rangle\varphi) = 
        \langle a\rangle\alpha(\varphi)$.
 \item $\alpha([\ct(a)]\varphi) = [a]\alpha(\varphi)$.
 \end{itemize}
\item The natural transformation 
 $\beta : \Mod_\mts \Rightarrow \Mod_\cc \comp\Phi$
 maps an MTS $(M,s)$ in $\Mod_\mts(A)$ to $(\calC(M),s)$, 
 and a morphism $R$ to itself.
\end{itemize}
\end{proposition}
\begin{proof}
For $A$ in $\Sign_\mts$, $(M,s)$ in $\Mod_\mts(A)$, and $\varphi$ in 
$\sen_\cc(\Phi(A))$, we prove the satisfaction condition
\[
(M,s)\models_\mts \alpha(\varphi) \iff \beta(M,s)\models_\cc\varphi
\]
by induction on $\varphi$.
The only non-trivial cases correspond to formulae of the form $\langle \cv(a) \rangle\varphi$ and
$[\ct(a)] \varphi$.
\begin{itemize}
\item For $\langle \cv(a)\rangle\varphi$, we reason thus:
\begin{eqnarray*}
(M,s) \models \alpha(\langle \cv(a)\rangle\varphi) 
&\iff &(M,s)\models \langle a\rangle\alpha(\varphi)\\
&\iff &\textrm{there is $s\tran{a}_\must p$ in $M$ 
               with $(M,p)\models\alpha(\varphi)$}\\
&\stackrel{\mathrm{IH}}{\iff} &
       \textrm{there is $s\tran{\cv(a)} p$ in $\calC(M)$
               with $(\calC(M),p)\models \varphi$}\\
&\iff &(\calC(M),s)\models \langle \cv(a)\rangle\varphi . 
\end{eqnarray*} 

\item For $[\ct(a)]\varphi$, we argue as follows:
\begin{eqnarray*}
(M,s) \models \alpha([\ct(a)]\varphi) 
&\iff &(M,s)\models [a]\alpha(\varphi)\\
&\iff &\textrm{$(M,p)\models\alpha(\varphi)$
               for all $s\tran{a}_\may p$ in $M$}\\
&\stackrel{\mathrm{IH}}{\iff} &
       \textrm{$(\calC(M),p)\models \varphi$
               for all $s\tran{\ct(a)} p$ in $\calC(M)$}\\
&\iff &(\calC(M),s)\models [\ct(a)]\varphi . 
\end{eqnarray*} 
This completes the proof. \qed
\end{itemize}  
\end{proof}
The importance of the above result is that MTSs modulo refinement and their
accompanying modal logic can be `translated in a truth preserving
fashion' into LTSs modulo the covariant-contravariant simulation
preorder and their companion modal logic. Now it is natural to ask
whether one can consider $\calI_\mts$ a `subinstitution' of
$\calI_\cc$. Once again there are several related notions of subinstitution,
but the minimum requirement that they all make is that the functor $\beta$,
which is used to translate the models between the institutions, is an
equivalence of categories.
Next we will show that even with a `natural' institution morphism from 
$\calI_\mts$ to $\calI_\cc$, it is not possible to present the former as a
subinstitutions of the latter by means of an embedding.

Recall that an object in a category is {\em weakly final} if any other
object has at least one arrow into it.

\begin{proposition}\label{final-object}
$\Mod_\cc(A,B,\emptyset)$ has weakly final objects but $\Mod_\mts(A)$ does 
not.  
\end{proposition}
\begin{proof}
First, consider the pair $(F,s)$ where $F$ is the LTS with a single
state $s$ and transitions $s\tran{a} s$ for every $a\in A$. (Note
that, if $A$ is empty, then $(F,s)$ is just the LTS $0$.)  It is
immediate to check that $(F,s)$ is a weakly final object of
$\Mod_\cc(A,B,\emptyset)$.

Now, assume that $(F',s')$ is weakly final in $\Mod_\mts(A)$ and consider 
the following two MTSs:
\begin{itemize}
\item $(M, m)$, with $m$ the only state in $M$ and transitions
 $m\tran{a}_\must m$ (and $m\tran{a}_\may m$) for every $a\in A$.
\item $(N, n)$, with $n$ the only state in $N$ and no transitions.
\end{itemize}
The existence of a morphism, that is a refinement, from $(M,m)$ to
$(F',s')$ implies that, for every $a\in A$, there must be transitions
of the form $s'\tran{a}_\must s'_a$ in $F'$ for some $s'_a$;
therefore, there are also transitions $s\tran{a}_\may s'_a$.  But
then, the morphism from $(N,n)$ to $(F',s')$ requires the existence of
transitions $n\tran{a}_\may n$ in $N$, which do not exist by the
definition of $N$.  Hence, there is no weakly final object in
$\Mod_\mts(A)$.  \qed
\end{proof}
%%%%%%%%%%%%%%%%%%%%%
In other words, in the absence of bivariant actions, there is a
universal implementation in the setting of LTSs modulo the
covariant-contravariant simulation preorder. Within that framework,
there is also a universal specification, namely the LTS $(I,s)$ where
$I$ is the LTS with a single state $s$ and transitions $s\tran{b} s$
for every $b\in B$. On the other hand, there is a universal
specification with respect to modal refinements, namely the MTS $U$
from Example~\ref{Ex:U}, but no universal implementation.

\begin{proposition}
There cannot exist an embedding $(\Phi,\alpha, \beta)$ from $\calI_\mts$ into 
$\calI_\cc$ such that $\Phi(A)$ does not have bivariant actions for some $A$.
\end{proposition}
\begin{proof}
If such an embedding existed then $\beta_A$, which is the natural
transformation translating MTSs into LTSs and refinement relations
into covariant-contravariant simulations, would be an equivalence
between $\Mod_\mts(A)$ and $\Mod_\cc(\Phi(A))$.  Since equivalences of
categories preserve weakly final objects, the result follows from
Proposition~\ref{final-object}.  \qed
\end{proof}
%%%%%%%%%%%%%%%%%%%
We will now argue that $\calI_\mts$ cannot be embedded into 
$\calI_\cc$ even in the presence of bivariant actions.
%%What if the embedding makes use of bivariant actions?
Recall that an object in a category is {\em weakly initial} if there
is at least one arrow from it into any other object.

\begin{proposition}\label{initial-object}
$\Mod_\mts(A)$ has weakly initial objects but $\Mod_\cc(A,B,C)$ does not if
$C\neq\emptyset$.  
\end{proposition}
\begin{proof}
Consider the MTS $(I,s)$ defined by $s\tran{a}_\may s$ for all $a\in
A$. We have already seen that it is weakly initial.

Now, assume that $(I',s')$ is weakly initial in $\Mod_\cc(A,B,C)$ and let
$c\in C$.
We define the following LTSs:
\begin{itemize}
\item $(P,p)$ with $p\tran{c} p$, and 
\item $(Q,q)$ with a single state $q$ and no transitions.
\end{itemize}
A morphism from $(I',s')$ to $(P,p)$ requires a transition $s'\tran{c}
s''$ in $I'$ for some $s''$. But then, a morphism from $(I',s')$ to
$(Q,q)$ requires a transition $q\tran{c} q$, which does not exist by
definition.  Therefore, $(I',s')$ cannot exist.  \qed
\end{proof}

\begin{proposition}
There cannot exist an embedding $(\Phi,\alpha, \beta)$ from $\calI_\mts$ into 
$\calI_\cc$ such that $\Phi(A)$ has bivariant actions for some $A$.
\end{proposition}
\begin{proof}
Such an embedding $\beta_A$, if it existed, would be an equivalence of categories between $\Mod_\mts(A)$ and 
$\Mod_\cc(\Phi(A))$. This cannot hold by Proposition~\ref{initial-object}
because equivalences of categories preserve weakly initial objects.
\qed
\end{proof}
%%%%%%%%%%%%%%%%%%%
By the way, next we also prove that there is no embedding in the reverse
direction, that is, from $\calI_\cc$ into $\calI_\mts$. 

\begin{proposition}\label{prop:no-embedding}
There exists no embedding from $\calI_\cc$ into $\calI_\mts$.
\end{proposition}
\begin{proof}
If such an embedding $(\Phi,\alpha, \beta)$ existed, 
$\beta_{(A,B,\emptyset)}$ would be an equivalence between
$\Mod_\cc(A,B,\emptyset)$ and $\Mod_\mts(\Phi(A,B,\emptyset))$, which
is not possible by Proposition~\ref{final-object} because equivalences
preserve weakly final objects.  \qed
\end{proof}
%%%%%%%%%%%%%%%%%%%%%%%
\label{instconj-start}
In \cite{AcetoEtAl10} we conjectured that there is not even an institution morphism
from $\calI_\cc$ to $\calI_\mts$; we now make this claim precise.

If we are not concerned about how contrived this morphism can be, then a 
trivial one can indeed be defined.
Let $\Phi$ map any signature to the singleton set $\{1\}$,
$\beta$ map any LTS to a MTS with a single state $s$ and transitions
$s\tran{1}_\may s$ and $s\tran{1}_\must s$, and $\alpha$ be recursively 
defined by $\alpha([1]\varphi)=\alpha(\varphi)$, $\alpha(\langle 1\rangle
\varphi) = \alpha(\varphi)$, and as expected in the remaining cases.
It is then a simple exercise to check that $(\Phi,\alpha,\beta)$
satisfies the conditions to be an institution morphism, however trivial and
artificial it may be.

Taking Proposition~\ref{inst-morphism:prop} as a model, and recalling the good 
properties of the function $\calM$ studied in Section~\ref{cc-to-modal:sec}, 
a ``natural'' morphism from $\calI_\cc$ to $\calI_\mts$ 
would be expected to satisfy $\beta(M,s) = (\calM(M),s)$.
We now argue that such a morphism cannot exist.

Assume that $(A,B,C)\in \Sign_\cc$, let $a\in A$ be any covariant action, and
$[a]\bot$ a Boudol-Larsen modal formula: how 
should $\alpha([a]\bot)$ be defined?
By the requirements of institution morphisms, the following equivalence
must hold for all LTS $M$:
\[
(M,s)\models_\cc \alpha([a]\bot) \iff \beta(M,s)\models_\mts [a]\bot.
\]
The right-hand side is true iff $(\calM(M),s')\models_\mts \bot$ for
all $s\tran{a}_\may s'$ in $\calM(M)$ which, by construction,
only holds if there is no $s'$ in $M$ with $s\tran{a} s'$ in $M$.
Therefore, $\alpha([a]\bot)$ has to be such that:
\begin{itemize}
\item $(M,s)\models_\cc \alpha([a]\bot)$ if there is no $s\tran{a} s'$ in $M$, 
 but
\item $(M,s)\not\models_\cc \alpha([a]\bot)$ if there is $s\tran{a} s'$ in $M$.
\end{itemize}
The immediate candidate would be $[a]\bot$ itself, now considered as a 
covariant-contravariant modal formula, but this is not possible since
in this framework the modality $[\_]$ requires a contravariant action.
Actually, no such formula can be defined which means that no institution
morphism with $\beta(M,s) = (\calM(M),s)$ can exist.
\label{instconj-end}
%%%%%%%%%%%%%%%%%%%%%%

Certainly, the given definitions of the institutions $\calI_\cc$ and
$\calI_\mts$ are not the only possible ones. We have also studied more general
institutions $\calI'_\cc$ and $\calI'_\mts$ where the functions $f$ in the
signature are replaced by relations but, unfortunately, the negative results
above remain valid.

\begin{itemize}
\item $\Sign'_\cc$ has as objects triples $(A,B,C)$ of pairwise disjoint sets
  and morphisms are relations $R\subseteq (A\times A')\cup (B\times B')\cup
  (C\times C')$. 

\item $\sen'_\cc(A,B,C)$ is the set of formulae in the logic characterizing the
  covariant-contravariant simulation preorder, with $A$ the set of covariant
  actions, $B$ the set of contravariant actions, and $C$ the set of bivariant
  actions.  
  For each morphism $R$ and formula $\varphi$, the formula
  $\sen'_\cc(R)(\varphi)$ is obtained from $\varphi$ by ``replacing'' 
  each action $a$ with every $a'$ such that $a R a'$. 
  More precisely, $\sen'_\cc(R)(\varphi)$ is defined recursively so that 
  $\langle a\rangle\varphi'$ becomes $\bigvee_{a R a'}\langle a'\rangle
  \sen'_\cc(R)(\varphi')$ and $[b]\varphi'$ becomes $\bigwedge_{b R b'}[b']
  \sen'_\cc(R)(\varphi')$.

\item $\Mod'_\cc(A,B,C)$ is the category of LTSs over the set of actions 
 $A\cup B\cup C$, with a distinguished state; a morphism from $(P,p)$ to
 $(Q,q)$ is a covariant-contravariant simulation $S$ such that $(p,q)\in S$.

 Now, if $R: (A,B,C) \lra (A',B',C')$ is a $\Sign'_\cc$-signature morphism, then
\[
\Mod'_\cc(R): \Mod'_\cc(A',B',C')\lra \Mod'_\cc(A,B,C)
\]
 maps $P$ to $R(P)$ and a simulation $S:P\lra Q$ to $R(S): R(P)\lra R(Q)$,
 where: 
 \begin{itemize}
 \item The set of states of $R(P)$ is the same as that of $P$, and the
   distinguished state remains the same.
 \item $p\tran{a} p'$ in $R(P)$ if $a R a'$ and $p\tran{a'} p'$ in $P$.
 \item $R(S)$ coincides with $S$.
 \end{itemize}

\item $(P,s)\models'_\cc\varphi$ if $(P,s)\models \varphi$ using the notion
 of satisfaction associated with the logic for the covariant-contravariant 
 simulation preorder given in Definition~\ref{Def:formulaeCC}.
\end{itemize}
That is, signature morphisms become arbitrary relations that `preserve' the
modality of the actions. 

Obviously, the institution $\calI_\mts$ could be subjected to an analogous
generalization; then, it would be a simple exercise to translate to this
new setting the results proved in
Propositions~\ref{prop:I_cc}--\ref{prop:no-embedding}.

\section{Conclusions and future work}\label{Sect:future}

In this paper we have studied the relationships between three notions
 of behavioural preorders that have been proposed in the literature:
 refinement over modal transition systems, and the
 covariant-contravariant simulation and the partial bisimulation
 preorders over labelled transition systems. We have provided mutual
 translations between modal transition systems and labelled transition
 systems that preserve, and reflect, refinement and the
 covariant-contravariant simulation preorder, as well as the the modal
 properties that can be expressed in the logics that characterize
 those preorders. We have also offered a translation from labelled
 transition systems modulo the partial bisimulation preorder into the
 same model modulo the covariant-contravariant simulation preorder,
 together with some evidence that the former model is less expressive
 than the latter. Finally, in order to gain more insight into the
 relationships between modal transition systems modulo refinement and
 labelled transition systems modulo the covariant-contravariant
 simulation preorder, we have also phrased and studied their
 connections in the very general abstract framework of institutions.

The work presented in the study opens several interesting avenues for
future research. Here we limit ourselves to mentioning a few research
directions that we plan to pursue in future work.

First of all, it would be interesting to study the relationships
between the LTS-based models we have considered in this article and
variations on the MTS model surveyed in, for
instance,~\cite{MTSBEATCS}. In particular, the third author recently
contributed in~\cite{FecherFLS09} to the comparison of several
refinement settings, including modal and mixed transition systems. The
developments in that paper offer a different approach to the
comparison and application of the formalisms studied in this article.

In~\cite{FabregasEtAl10-logics}, three of the authors gave a
ground-complete axiomatization of the covariant-contravariant
simulation preorder over the language BCCS~\cite{Mi89}.  It would be
interesting to see whether the translations between MTSs and LTSs we
have provided in this paper can be used to lift that axiomatization
result, as well as results on the nonexistence of finite
(in)equational axiomatizations, to the setting of modal transition
systems modulo refinement, using the BCCS-like syntax for MTSs given
in~\cite{BoudolL1992} and used in Section~\ref{Sect:charforms} of this
paper. We also intend to study whether our translations can be used
to obtain characteristic-formula
constructions~\cite{BoudolL1992,GrafS86a,SteffenI94} for one model
from extant results on the existence of characteristic formulae for
the other. In the setting of the finite LTSs that are the image of MTS
terms via $\calC$, this has been achieved in
Section~\ref{Sect:charforms} of this study.

The existence of characteristic formulae allows one to reduce checking
the existence of a behavioural relation between two processes to a
model checking question. Conversely, the main result
from~\cite{BoudolL1992} offers a complete characterization of the
model checking questions of the form $(M,m)\models \varphi$, where $M$
is an MTS and $\varphi$ is a formula in the logic for MTSs considered
in this paper, that can be reduced to checking for the existence of a
refinement between $(M_\varphi,m_\varphi)$ and $(M,m)$, where
$(M_\varphi,m_\varphi)$ is an MTS with a distinguished state that
`graphically represents' the formula
$\varphi$. In~\cite{AcetoEtAl11}, we offered a characterization of
the logical specifications that can be `graphically represented' by
LTSs modulo the covariant-contravariant simulation preorder and
partial bisimilarity. This result applies directly to LTSs whose
signature contains no bivariant actions. Such a characterization may
shed further light on the relative expressive power of the two
formalisms and may give further evidence of the fact that LTSs modulo
the covariant-contravariant simulation preorder are, in some suitable
formal sense, more expressive than LTSs modulo partial bisimilarity.

\iffalse
From the theoretical point of view, it would also be satisfying to
settle our conjecture that there is no institution morphism from
$\calI_\cc$ to $\calI_\mts$.
\fi

Last, but not least, the development of the notion of partial
bisimulation in~\cite{Baetenetal,Baetenetal1} has been motivated by
the desire to develop a process-algebraic model within which one can
study topics in the field of {\em supervisory
control}~\cite{RW87}. Recently, MTSs have been used as a suitable
model for the specification of service-oriented applications, and
results on the supervisory control of systems whose specification is
given in that formalism have been presented in,
e.g.,~\cite{DaroundeauDM2010,FP2007}. It is a very interesting area
for future research to study whether the mutual translations between
MTSs modulo refinement and LTSs modulo the covariant-contravariant
simulation preorder can be used to transfer results on supervisory
control from MTSs to LTSs. One may also wish to investigate directly
the adaptation of the supervisory control theory of Ramadge and Wonham
to the enforcement of specifications given in terms of LTSs modulo the
covariant-contravariant simulation preorder.

\paragraph{Acknowledgments} 
We thank the three anonymous reviewers and the guest editors for their
insightful comments that led to improvements in the paper.

%% \bibliographystyle{abbrv}
%% \bibliography{abbreviations,mts-cc,proceedings}
\def\polhk#1{\setbox0=\hbox{#1}{\ooalign{\hidewidth
  \lower1.5ex\hbox{`}\hidewidth\crcr\unhbox0}}}

\end{document}